\documentclass[11pt]{article}

\usepackage{amsfonts}

\usepackage{multirow}
\usepackage{enumerate}
\usepackage{amsmath, dsfont, amscd, latexsym}
\usepackage[usenames,dvipsnames]{color}
\usepackage{amssymb}
\usepackage{setspace}
\usepackage{ifthen} 
\usepackage{graphicx}
\usepackage[titles]{tocloft} 

\usepackage{hyperref} 
\hypersetup{colorlinks=true, linkcolor=blue, citecolor=cyan}

\newtheorem{theorem}{Theorem}
\newtheorem{corollary}[theorem]{Corollary}
\newtheorem{conjecture}[theorem]{Conjecture}

\newtheorem{lemma}[theorem]{Lemma}
\newtheorem{example}[theorem]{Example}

\newtheorem{prop}[theorem]{Proposition}
\newtheorem{defn}[theorem]{Definition}
 
\newtheorem{fact}[theorem]{Fact}

\newenvironment{proof}{\noindent
  \textit{Proof. }\nopagebreak[2]}{$\qed$}

\usepackage{supertabular} 
\usepackage{verbatim}

\usepackage{pstricks,pst-node,pst-tree}
\usepackage{boxedminipage}
\newcommand{\qed}{\hfill\rule{7pt}{7pt} \medskip}

\newcommand{\eps}{\varepsilon}

\newcommand{\sgn}{\mathrm{sgn}}
\newcommand{\note}[1]{\marginpar{\tiny *note in TeX*}}

\newcommand{\ignore}[1]{}

\newcommand{\cref}[1]{Corollary~\ref{cor:#1}}

\newcommand{\set}[1]{\left\{#1\right\}}

\newcommand{\bits}{\{-1,1\}}
\newcommand{\zbits}{\{0,1\}}
 
\newcommand{\bn}{\bits^n}
\newcommand{\isafunc}{: \bn \rightarrow \bits}

\newcommand{\R}{\mathbb{R}}

\newcommand{\N}{\mathbb{N}}
\newcommand{\F}{\mathbb{F}}

\newcommand{\Var}{\operatorname{{\bf Var}}}

\newcommand{\E}{\operatorname{{\bf E}}}
\newcommand{\Ex}{\mathop{{\bf E}\/}}
\renewcommand{\Pr}{\operatorname{{\bf Pr}}}
\newcommand{\Prx}{\mathop{{\bf Pr}\/}}
\newcommand{\lsum}{\mathop{\textstyle \sum}}
\newcommand{\lprod}{\mathop{\textstyle \prod}}
\newcommand{\lfrac}{\textstyle \frac}
\newcommand{\poly}{\mathrm{poly}}

\newcommand{\Inf}{\mathrm{Inf}}

\newcommand{\half}{\frac{1}{2}}
\newcommand{\wh}{\widehat}

\newcommand{\la}{\langle}
\newcommand{\ra}{\rangle}

\newcommand{\st}{\,|\,} 
\newcommand{\stc}{\,:\,} 
\newcommand{\inv}[1]{\frac{1}{#1}}
\newcommand{\sumi}{\sum_{i=1}^n} 
\newcommand{\norm}[1]{\|#1\|}

\newcommand{\cals}{\mathcal{S}}

\newcommand{\calh}{\mathcal{H}}
\newcommand{\calu}{\mathcal{U}}
\newcommand{\calg}{\mathcal{G}}

\newcommand{\bfX}{\mathbf{X}}
\newcommand{\bfY}{\mathbf{Y}}
\newcommand{\bfZ}{\mathbf{Z}}
\newcommand{\bfR}{\mathbf{R}}
\newcommand{\bfS}{\mathbf{S}}

\newcommand{\sse}{\subseteq}

\newcommand{\chis}{\chi_S}
\newcommand{\sumalls}{\sum_{S\subseteq [n]}} 
\newcommand{\hatf}{\hat{f}}
\newcommand{\hatfs}{\hat{f}(S)}

\newcommand{\NAE}{\mathsf{NAE}}
\newcommand{\BLR}{\mathsf{BLR}}

\newcommand{\accepts}{\mathrm{accepts}}

\newcommand{\opt}{\mathsf{opt}}

\newcommand{\MAJ}{\mathsf{MAJ}}

\newcommand{\sfand}{\mathsf{AND}}
\newcommand{\sfor}{\mathsf{OR}}

\newcommand{\indi}{{\bf 1}}

\newcommand{\dist}{\mathrm{dist}}

\newcommand{\wt}{\mathrm{wt}} 
\newcommand{\error}{\mathrm{error}}

\newcommand{\bprop}{\begin{prop}}
\newcommand{\eprop}{\end{prop}}
\newcommand{\bcor}{\begin{corollary}} 
\newcommand{\ecor}{\end{corollary}} 
\newcommand{\bthm}{\begin{theorem}}
\newcommand{\ethm}{\end{theorem}}
\newcommand{\blem}{\begin{lemma}}
\newcommand{\elem}{\end{lemma}} 
\newcommand{\bpf}{\begin{proof}}
\newcommand{\epf}{\end{proof}}
\newcommand{\beg}{\begin{example}} 
\newcommand{\eeg}{\end{example}} 
\newcommand{\bdefn}{\begin{defn}}
\newcommand{\edefn}{\end{defn}}
\newcommand{\bit}{\begin{itemize}} 
\newcommand{\eit}{\end{itemize}} 
\newcommand{\ben}{\begin{enumerate}}
\newcommand{\een}{\end{enumerate}} 

\newcommand{\bbox}{\begin{mybox}}
\newcommand{\ebox}{\end{mybox}} 

\newcommand{\bquote}{\begin{quote}}
\newcommand{\equote}{\end{quote}} 

\newenvironment{mybox}
{\begin{center}
\begin{boxedminipage}{5.3in}
\begin{minipage}{5in}
\vspace{6pt}
} 
{\vspace{.1pt}
\vspace{-2pt}
\end{minipage}
\end{boxedminipage}
\end{center}}

\advance \topmargin by -\headheight
\advance \topmargin by -\headsep
\evensidemargin \oddsidemargin
\begin{document}
\thispagestyle{empty}
\begin{center}
\vspace*{1.5in}
\Huge{\bf Analysis of Boolean Functions}\\
\vspace{.3in}
\LARGE{Notes from a series of lectures by} \\
\vspace{.1in} 
\huge{Ryan O'Donnell} \\
\vspace{.2in}
\LARGE{Guest lecture by Per Austrin} 
\vspace{1in} \\ 
\LARGE{Barbados Workshop on Computational Complexity} \\
\LARGE{February 26th -- March 4th, 2012} \\ 
\LARGE{Organized by Denis Th\'erien}  \\ 
\LARGE{Scribe notes by Li-Yang Tan} 
\end{center}

 \pagebreak 
\vspace*{\fill}
\setcounter{tocdepth}{2} 
\tableofcontents 
\vspace*{\fill}
\thispagestyle{empty}
\pagebreak 

\newcommand{\bfW}{\mathbf{W}}
\newcommand{\Stab}{\mathrm{Stab}}
\newcommand{\lhalf}{{\lfrac 1 2}}
\newcommand{\GStab}{\mathrm{GStab}}
\newcommand{\RS}{\mathrm{RS}}
\newcommand{\sqdist}{\mathsf{sqdist}}

\section{Linearity testing and Arrow's theorem} 

\begin{center}
\large{Monday, 27th February 2012} 
\end{center}

\bit
\itemsep -.5pt   
\item Open Problem \cite{Guy86, HK92}: Let $a\in\R^n$ with $\norm{a}_2
  = 1$. Prove $\Pr_{x\in\bn}[|\la a,x\ra|\leq 1] \geq
  \half$. 
\item Open Problem (S. Srinivasan): Suppose
  $g:\bn\to\pm\big[\frac{2}{3},1\big]$ where $g(x) \in
  \big[\frac{2}{3},1\big]$ if $\sumi x_i\geq \frac{n}{2}$ and $g(x)
  \in \big[-1,-\frac{2}{3}\big] $ if $\sumi x_i \leq
  -\frac{n}{2}$. Prove $\deg(f) = \Omega(n)$.  \eit\medskip

  In this workshop we will study the analysis of boolean functions and
  its applications to topics such as property testing, voting,
  pseudorandomness, Gaussian geometry and the hardness of
  approximation. Two recurring themes that we will see throughout the
  week are:

\bit
\itemsep -.5pt
\item The noisy hypercube graph is a small set expander. 
\item Every boolean function has a ``junta part'' and a ``Gaussian
  part''. 
\eit 

\subsection{The Fourier expansion}
 
Broadly speaking, the analysis of boolean functions is concerned with
properties of boolean functions $f\isafunc$ viewed as multilinear
polynomials over $\R$.  Consider the majority function over 3
variables $\MAJ_3(x) = \sgn(x_1 + x_2 + x_3)$.  It is easy to check
that $\MAJ_3(x) = \half x_1 + \half x_2 + \half x_3 - \half
x_1x_2x_3$, and this can be derived by summing $2^3 = 8$ polynomials
$p_y:\bits^3 \to\{-1,0,1\}$, one for each $y\in \bits^3$, where
$p_y(x)$ takes value $\MAJ_3(x)$ when $y = x$ and $0$ otherwise.  For
example,
\[ p_{(-1,1,-1)}(x) =
\Big(\frac{1-x_1}{2}\Big)\Big(\frac{1+x_2}{2}\Big)\Big(\frac{1-x_3}{2}
\Big) \cdot \MAJ_3(-1,1,-1). \] Note that the final polynomial that
results from expanding and simplifying the sum of $p_y$'s is indeed
always multilinear ({\it i.e.} no variable $x_i$ occurs squared, or
cubed, {\it etc.})  since $x_i^2 = 1$ for bits $x_i \in \bits$.  The
same interpolation procedure can be carried out for any $f:\bn\to\R$: 

\bthm[Fourier expansion]
\label{bthm:fourierexpansion}
Every $f:\bn\to\R$ can be uniquely expressed as a multilinear
polynomial $\R$, 
\[ f(x) = \sumalls c_S \prod_{i\in S} x_i, \quad \text{where each
  $c_S\in \R$.} \] We will write $\hatfs$ to denote the coefficient
$c_S$ and $\chis(x)$ for the function $\prod_{i \in S} x_i$, and call
$f(x) = \sumalls\hatfs\chis(x)$ the Fourier expansion of $f$.  We
adopt the convention that $\chi_\emptyset \equiv 1$, the identically
$1$ function.  We will write $\deg(f)$ to denote $\max_{S\sse
  [n]}\{|S|\stc \hatfs\neq 0\}$, and call this quantity the Fourier
degree of $f$. \ethm

We will sometimes refer to $\chis(x)\isafunc$ as the ``parity-on-$S$''
function, since it takes value 1 if there are an even number of $-1$
coordinates in $x$ and $-1$ otherwise.  Using the notation of Theorem
\ref{bthm:fourierexpansion}, we have that $\wh{\MAJ_3}(\set{1}) =
\half$, $\wh{\MAJ_3}(\set{1,2,3}) = -\half$, $\wh{\MAJ_3}(\set{1,2})=
0$, and $\deg(\MAJ_3) = 3$. \medskip

We have already seen that every function $f:\bn\to\R$ can be expressed
as a multilinear polynomial over $\R$ (via the interpolation procedure
described for $\MAJ_3$); to complete the proof of Theorem
\ref{bthm:fourierexpansion} it remains to show uniqueness.  Let $V$ be
the vector space of all functions $f:\bn\to\R$.  Here we are viewing
$f$ as a $2^n$-dimensional vector in $\R^{2^n}$, with each coordinate
being the value of $f$ on some input $x\in\bn$; if $f$ is a boolean
function this is simply the truth table of $f$.  Note that the parity
functions $\chis(x)$ are all elements of $V$, and furthermore every
$f\in V$ can be expressed as a linear combination of them ({\it i.e.}
the $\chis(x)$'s are a spanning set for $V$). Since there are $2^n$
parity functions and $\dim(V)= 2^n$, it follows that $\set{\chis\stc
  S\sse [n]}$ is a basis for $V$, and this establishes the uniqueness
of the Fourier expansion.

\bdefn[inner product]
Let $f,g\isafunc$. We define the inner product between $f$ and $g$ as 
\[ \la f,g\ra := \sum_{x\in\bn}\frac{f(x) \cdot g(x)}{2^n} =
\Ex_{x\in\bn}[f(x)g(x)].\] \edefn Note that this is simply the dot
product between $f$ and $g$ viewed as vectors in $\R^{2^n}$,
normalized by a factor of $2^{-n}$.  Given this definition, every
boolean function $f\isafunc$ is a unit vector in $\R^{2^n}$ since $\la
f,f\ra 
= 1$. We will also write $\norm{f}_2^2$ to denote $\la f,f\ra$, and
more generally, $\norm{f}_p := \E[|f(x)|^p]^{1/p}$.

\bthm[orthonormality]
\label{bthm:orthonormality} 
The set of parity functions $\set{\chis(x)\stc S\sse [n]}$ is an
orthonormal basis for $\R^{2^n}$.  That is, for every $S,T\sse [n]$, 
\[ \la \chis,\chi_T\ra = \left\{
\begin{array}{ll}
1 & \text{if $S = T$} \\
0 & \text{otherwise.}  
\end{array}
\right.
\] 
\ethm

\bpf
First note that $\chis\cdot \chi_T = \chi_{S\Delta T}$ since $\prod_{i
\in S} x_i \prod_{j \in T} x_j = \prod_{i \in S\Delta T}x_i \prod_{j
\in S\cap T}x_j^2 = \prod_{i \in S\Delta T} x_i$, where the final
equality uses the fact that $x_i^2 = 1$ for $x_i \in \bits$. Next, we
claim that  
\[ \E[\chi_U] = \left\{
\begin{array}{ll}
1 & \text{if $U = \emptyset$}\\ 
0 & \text{otherwise.}  
\end{array} 
\right. 
\] 
noting that this implies the theorem since $S\Delta T = \emptyset$ iff
$S = T$. Recall that we have defined $\chi_\emptyset$ to be the
identically 1 function, and if $U\neq\emptyset$ then exactly half the
inputs $x \in \bn$ have $\chi_U(x) = 1$ and the other half $\chi_U(x) =
-1$.  \epf

\bprop[Fourier coefficient]
Let $f:\bn\to\R$.  Then $\hatfs = \la f,\chis\ra = \E[f(x)\chis(x)]$. 
\eprop 

\bpf 
To see that this holds, we check that 
\[ \la f,\chis \ra = \Big\la \lsum_{T\sse
  [n]}\hatf(T)\chi_T,\chis\Big\ra = \sum_{T\sse [n]}\hatf(T)\cdot
\la\chi_T,\chis\ra = \hatfs. \] Here we have used the Fourier
expansion of $f$ for the first equality, linearity of the inner
product for the second, and orthonormality of parity functions
(Theorem \ref{bthm:orthonormality}) for the last.\epf

Next we have Plancherel's theorem, which states that the inner product
of $f$ and $g$ is precisely the dot product of their vectors of
Fourier coefficients. 
\begin{theorem}[Plancherel]
Let $f,g:\bn\rightarrow\R$. Then, $\la f,g \ra =  \sumalls
\hat{f}(S)\hat{g}(S).$ 
\end{theorem} 

\begin{proof}
Again we use the Fourier expansions of $f$ and $g$ to check that 
  \[ \la f,g\ra = \Big\la \lsum_{S\sse [n]} \hatfs\chi_S, \lsum_{T\subseteq
    [n]}\hat{g}(T)\chi_T\Big\ra = \sum_S\sum_T \hatfs\hat{g}(T)\cdot
  \la\chi_S,\chi_T\ra = \sumalls \hatfs\hat{g}(S). \] The second
  equality holds by linearity of inner product, and the last by
  orthonormality.
\end{proof}

An important corollary of Plancherel's theorem is Parseval's identity:
if $f:\bn\to\R$, then $\norm{f}_2^2 = \la f,f\ra = \sumalls \hatfs^2$
({\it i.e.} the Fourier transform preserves $L_2$-norm).  In
particular, if $f$ is a boolean function then $\sumalls \hatfs^2 =
\E[f(x)^2] = 1$, which we may view as a probability distribution over
the $2^n$ possible subsets $S$ of $[n]$.  Note that if $f$ and $g$ are
boolean functions then $f(x)g(x) = 1$ iff $f(x) = g(x)$, and so
$\E[f(x)g(x)] = 1-2\cdot \dist(f,g)$, where $\dist(f,g) = \Pr[f(x)\neq
g(x)]$ is the normalized Hamming distance between $f$ and
$g$. \medskip

One of the advantages of analyzing $f$ via its Fourier expansion is
that this polynomial encodes a lot of combinatorial information about
$f$, and these combinatorial quantities can be ``read off'' its
Fourier coefficients easily.  We give two basic examples now. Recall
that for functions $f:\bn\to\R$, the mean of $f$ is $\E[f(x)]$ and its
variance is $\Var(f) := \E[f(x)^2] - \E[f(x)]^2$. Note that if $f$ is
a boolean function then $\E[f(x)]$ measures the bias of $f$ towards
$1$ or $-1$, and $\Var(f) = 4\cdot \Pr[f(x) = 1] \cdot \Pr[f(x) =
-1]$. If $f$ has mean 0 and variance 1 we say that $f$ is balanced, or
unbiased.

\bprop[expectation and variance]
\label{prop:bexandvar}
Let $f:\bn\to\R$.  Then $\E[f(x)] = \hatf(\emptyset)$, and $\Var(f) =
\sum_{S\neq\emptyset}\hatfs^2$. 
\eprop 

\bpf 
For the first equality, we check that $\hatf(\emptyset) =
\E[f(x)\chi_\emptyset(x)] = \E[f(x)]$.  The second equality holds
because 
\[ \Var(f) = \E[f(x)^2] - \E[f(x)]^2 = \Big(\lsum_{S\sse [n]}\hatfs^2\Big) -
\hatf(\emptyset)^2 = \sum_{S\neq\emptyset}\hatfs^2. \] 
Here the second equality uses an application of Parseval's identity. 
\epf 

It is nice to think of $\hatfs^2$ as the ``weight'' of $f$ on $S$,
with the sum of weights of $f$ on all $2^n$ subsets $S$ of $[n]$ being
1 by Parseval's.  Often it will also be convenient to stratify these
weights according to the cardinality of the set $S$.

\bdefn[level weights] Let $f:\bn\to\R$ and $k\in
\set{0,\ldots,n}$. The weight of $f$ at level $k$, or the degree-$k$
weight of $f$, is defined to be $\bfW^k(f) := \sum_{|S|=k}\hatfs^2$.  
 \edefn

For example, in this notation we have $\bfW^0(\MAJ_3) =
\bfW^2(\MAJ_3) = 0$ and $\bfW^1(\MAJ_3) = \bfW^3(\MAJ_3) = 1/2$. 

\subsubsection{Density functions and convolutions}  

So far we have been viewing domain $\bn$ of our functions simply as
strings of bits represented by real numbers $\pm 1$. Often we
would like to be able to ``add'' two inputs, in which case we will
view our functions as $f:\F_2^n\to\R$ instead. The mapping from $\F_2$
to $\bits$ is given by $(-1)^b$, sending $0\in \F_2$ to $1\in \R$ and
$1\in\F_2$ to $-1\in\R$.  We will sometimes also associate a boolean
function $\F_2^n\to\bits$ with its corresponding $\F_2$ polynomial
$\F_2^n\to\F_2$.  For example, the different representations of the
parity function $\chis$ are given in Table
\ref{table:parity}. \medskip 
 
\begin{table}[h]
\renewcommand{\arraystretch}{1.4} 
\begin{center}
\begin{tabular}{|c|c|}
\hline
$\chis(x)\isafunc$ & $x \mapsto \prod_{i \in S} x_i$ \\
\hline
$\chis(x):\F_2^n \to \bits$ & $x \mapsto \prod_{i\in S}(-1)^{x_i}$ \\ 
\hline
$\chis(x) :\F_2^n\to\F_2$ & $x \mapsto \sum_{i\in S} x_i$ \\
\hline
\end{tabular}
\end{center}
\vspace{-0.1in}
\caption{Different representations of $\chis(x)$} 
\label{table:parity}
\end{table}
 
The $\F_2$ degree of a boolean function $f:\F_2^n\to\bits$, denoted
$\deg_{\F_2}(f)$, is its degree as an $\F_2$ polynomial
$\F_2^n\to\F_2$. For example, the parity functions are degree-1 $\F_2$
polynomials; in contrast, recall that $\deg(\chis)$, the Fourier
degree of $\chis$, is $|S|$.  In general we have the inequality
$\deg_{\F_2}(f)\leq \deg(f)$ for all boolean functions $f
:\F_2^n\to\bits$. We remark that unlike Fourier degree, it is not
known how to infer the $\F_2$ degree of a boolean function from its
Fourier expansion. The following useful fact can be easily verified:

\begin{fact}
\label{fact:chihomo}
Let $\chis:\F_2^n\to\bits$. Then $\chis(x+y) = \chis(x)\cdot
\chis(y)$. 
\end{fact}

\bdefn[probability density function]
$\varphi: \F_2^n\to\R^{\geq 0}$ is a probability density function of
$\Ex_{x\in\F_2^n}[\varphi(x)] = 1$. 
\edefn

Note that a probability density function $\varphi : \F_2^n\to\R^{\geq
  0}$ corresponds to the probability distribution over $\F_2^n$ where
$\Pr[x] = \varphi(x)\cdot 2^n$.  For example, the constant function
$\varphi \equiv 1$ corresponds to the uniform distribution over
$\F_2^n$. For any $a\in\F_2^n$, the density function
$\varphi_a(x)$ that takes value $2^n$ if $x=a$ and $0$ otherwise
corresponds to the distribution that puts all its weight on a single
point $a\in \F_2^n$. 

\bdefn[convolution]
Let $f,g:\F_2^n\to\R$. The convolution of $f$ and $g$ is the function
$f*g:\F_2^n\to\R$ defined by $(f*g)(x) :=
\E_{y\in\F_2^n}[f(y)g(x+y)]$. 
\edefn

Note that $(f*g)(x) = (g*f)(x)$, since
$(\boldsymbol{y},\boldsymbol{y}+x)$ is just a uniformly random pair of
inputs with distance $x$ and therefore has the same distribution as
$(\boldsymbol{y}+x,\boldsymbol{y})$. Similarly it can be checked that
the convolution operator is commutative: $(f*g)*h = f*(g*h)$.  The
following facts also follow easily from definitions: 

\begin{fact} Let $f:\F_2^n\to\R$ and $\varphi_2, \varphi_2$ be density
  functions.  Then
\ben
\itemsep -.5pt
\item $\la \varphi, f\ra = \Ex_{y\sim\varphi}[f(y)]$. 
\item $(\varphi * f)(x) = \E_{y\sim\varphi}[f(x+y)]$.
\item The density for $z = y_1 + y_2$, where $y_1\sim \varphi_1$ and
  $y_2\sim\varphi_2$, is $\varphi_1 * \varphi_2$. 
\een  
\end{fact}

\bthm[Fourier coefficients of convolutions]
\label{thm:bcoeffofconvo}
Let $f,g:\F_2^n\to\R$. Then $\wh{f*g}(S) = \hatfs\cdot \hat{g}(S)$. 
\ethm

\bpf 
We check that 
\begin{eqnarray}
\wh{f*g}(S) \ = \ \Ex_x[(f*g)(x)\chis(x)] 
&=&  \Ex_x\Big[\Ex_y[f(y)g(x+y)]\cdot \chis(x)\Big] \nonumber \\
&=&  \Ex_y\Big[f(y)\cdot \Ex_x[g(x+y)\chis(x)]\Big] \nonumber \\
&=& \Ex_y\Big[f(y)\cdot \Ex_z[g(z)\chis(z+y)]\Big]\label{eq:beq1} \\
&=& \Ex_y\Big[f(y)\cdot
\Ex_z[g(z)\chis(z)\chis(y)]\Big]\label{eq:beq2} \\
&=& \Ex_y[f(y)\chis(y)]\cdot \Ex[g(z)\chis(z)] \ = \ \hatfs\cdot
\hat{g}(S). \nonumber
\end{eqnarray} 
Here (\ref{eq:beq1}) uses the fact that $z-y = z+y$ for
$y,z\in\F_2^n$, and (\ref{eq:beq2}) is an application of Fact
\ref{fact:chihomo}.
\epf

\bthm Let $f,g,h:\F_2^n\to\R$. Then $\la f*g,h\ra = \la f,g*h\ra$.
\ethm \bpf By Theorem \ref{thm:bcoeffofconvo} and Plancherel, both
sides of the identity equal $\sumalls \hatfs\hat{g}(S)\hat{h}(S)$.
\epf

\subsection{Blum-Luby-Rubinfeld}
  
We begin by considering two notions of what it means for a function
$f:\F_2^n\to \F_2$ to be linear.

\bdefn[linear $\#1$] 
\label{def:blin1}
A boolean function $f:\F_2^n\to\F_2$ is linear if $f(x+y) =
f(x)+f(y)$ for all $x,y\in\F_2^n$. 
\edefn

\bdefn[linear $\#2$]
\label{def:blin2} 
A boolean function $f:\F_2^n\to\F_2$ is linear if there exists
$a_1,\ldots,a_n\in\mathbb{F}_2$ such that
$f(x)=a_1x_1+\ldots+a_nx_n$. Equivalently, there exists some $S\sse
[n]$ such that $f(x) = \sum_{i\in S}x_i$.  \edefn

\bprop[$\#1 \Longleftrightarrow \#2$]
\label{prop:linequiv}
These two definitions are equivalent. 
\eprop  

\bpf Suppose $f$ satisfies $f(x+y)=f(x)+f(y)$ for all $x,y\in\F_2^n$.
Let $\alpha_i = f(e_i)\in\F_2$, where $e_i$ is the $i$-th canonical
basis vector for $\mathbb{F}_2^n$.  It follows that $f(x) = f(\sumi
x_ie_i) = \sumi x_if(e_i) = \sumi \alpha_ix_i$, where the second
equality uses Definition \ref{def:blin1} repeatedly, along with the
fact that $f(x_ie_i) = x_if(e_i)$.  For the reverse implication, note
that $f(x+y) = \sum_{i\in S}(x+y)_i = \sum_{i\in S}x_i + \sum_{i\in
  S}y_i = f(x)+f(y)$, where the first and final equalities uses
Definition \ref{def:blin2}.  \epf

It is natural to consider analogous notions for approximate
linearity. 

\bdefn[approximately linear $\#1$]
\label{def:bapproxlin1}
A boolean function $f:\F_2^n\to\F_2$ is approximately linear if
$f(x+y)= f(x)+f(y)$ for most pairs $x,y\in\F_2^n$. 
\edefn

\bdefn[approximately linear $\#2$]
\label{def:bapproxlin2}
A boolean function $f:\F_2^n\to\F_2$ is approximately linear if there
exists some $S\sse [n]$ such that $f(x) = \sum_{i\in S}x_i$ for most
$x\in\F_2^n$. Equivalently, there exists an $S\sse [n]$ such that $f$
is close in Hamming distance to $g(x) = \sum_{i\in S} x_i$. 
\edefn
 
A straightforward generalization of argument given in the proof of
Proposition \ref{prop:linequiv} shows that Definition
\ref{def:bapproxlin2} (approximately linear $\#2$) implies Definition
\ref{def:bapproxlin1} (approximately linear $\#1$). However, the
argument for the reverse implication no longer holds. We will adopt
Definition \ref{def:bapproxlin2} as our notion of approximate
linearity for now, and we will see that the linearity test of Blum,
Luby, and Rubinfeld \cite{BLR93} implies that both definitions are in
fact equivalent.  The Fourier-analytic proof we present here is due to
Bellare {\it et.\,al} \cite{BCH+96}.

\bdefn[BLR linearity test]
Given blackbox access to a function $f:\F_2^n\to \F_2$, 
\ben
\itemsep -.5pt
\item Pick $x,y\in\F_2^n$ independently and uniformly. 
\item Query $f$ on $x,y$ and $x+y$. 
\item Accept iff $f(x) + f(y) = f(x+y)$. 
\een 
\edefn

\bthm[soundness of $\BLR$] 
\label{thm:bblrsoundness}
If $\Pr[\BLR\ \accepts\ f] \geq 1-\eps$ then
$f$ is $\eps$-close to being linear (in the sense of Definition
\ref{def:bapproxlin2}).  
\ethm

\bpf It will be convenient to think of $f$ as $\F_2^n\to\bits$, and so
the acceptance criterion ({\it i.e.} step 3) becomes $f(x) f(y) =
f(x+y)$.  Viewing $f$ this way, now note that 
\begin{eqnarray}
\Pr[\BLR\ \accepts\ f] & = & \Ex_{x,y}\big[\indi(f(x)\cdot f(y) =
f(x+y)\big] \nonumber \\  &=&  \Ex_{x,y}\big[{\lfrac 1 2} + {\lfrac 1
  2}\cdot  f(x)f(y)f(x+y)\big] \nonumber\\ 
&=& {\lfrac 1 2} + {\lfrac 1 2} \Ex_{x}[f(x)\cdot (f*f)(x)] \nonumber \\
&=& {\lfrac 1 2} + {\lfrac 1 2}\sumalls \hatfs\cdot \wh{f*f}(S) \label{eq:bblr1}\\
& =&  {\lfrac 1 2} + {\lfrac 1 2} \sumalls \hatfs^3. \label{eq:bblr2} 
\end{eqnarray} 
Here (\ref{eq:bblr1}) uses Parseval's identity and (\ref{eq:bblr2}) uses
Theorem \ref{thm:bcoeffofconvo}. Therefore, if $\Pr[\BLR\ \accepts\
f]\geq 1-\eps$ then $1-2\eps \leq \sumalls \hatfs^3 \leq
\max\{\hatfs\}\cdot \sumalls \hatfs^2 = \max\{\hatfs\}$, or
equivalently, there exists an $S^*\sse [n]$ such that $\hatf(S^*)\geq
1-2\eps$. Since $\hatf(S^*) = \E[f(x)\chi_{S^*}(x)] = 1-2\cdot
\dist(f,\chi_{S^*})$, we have shown that $\dist(f,\chi_{S^*})\leq
\eps$ and the proof is complete.  
\epf

Theorem \ref{thm:bblrsoundness} says that if $\Pr[\BLR\ \accepts\
f]\geq 1-\eps$ then $f$ is $\eps$-close to some linear function
$\chi_{S^*}$; however, we do not know \emph{which} of the $2^n$
possible linear functions $\chi_{S^*}$ this is. The following theorem
tells us that we can nevertheless obtain the correct value of
$\chi_{S^*}(x)$ with high probability for all $x\in\F_2^n$.

\bthm[local decodability of linear functions] 
Let $f:\F_2^n\to\F_2$ be $\eps$-close to some linear function
$\chi_{S^*}$, and let $x\in\F_2^n$. The following algorithm outputs
$\chi_{S^*}(x)$ with probability at least $1-2\eps$: 
\ben
\itemsep -.5pt
\item Pick $y\in\F_2^n$ uniformly. 
\item Output $f(y) + f(x+y)$. 
\een 
\ethm

\bpf Since $x$ and $x+y$ are both uniform (though not independent),
with probability at least $1-2\eps$ we have $f(y) = \chi_{S^*}(y)$ and
$f(x+y) = \chi_{S^*}(x+y)$.  The claim follows by noting that
$\chi_{S^*}(y) + \chi_{S^*}(x+y) = \chi_{S^*}(y+(x+y)) =
\chi_{S^*}(x)$, where we have used the linearity of $\chi_{S^*}$ along
with the fact that $x+y = x-y$ for $x,y\in\F_2^n$. \epf

\subsection{Voting and influence} 

\bit
\itemsep -.5pt
\item Puzzle: Is it possible for $f\isafunc$ to have exactly $k$
  non-zero Fourier coefficients, for $k = 0,1,2,3,4,5,6,7$? Classify all
  functions with $2$ non-zero Fourier coefficients.
\item Puzzle: Find all $f\isafunc$ with $\bfW^1(f) = 1$. 
\eit

We may think of a boolean function $f\isafunc$ as a voting scheme for
an election with 2 candidates ($\pm 1$) and $n$ voters
($x_1,\ldots,x_n$). Many boolean functions are named after the voting
schemes they correspond to: the $i$-th dictator ${\sf DICT}_i(x) =
x_i$ ({\it i.e.}  ${\sf DICT}_i \equiv \chi_i$); $k$-juntas (functions
that depend only on $k$ of its $n$ variables, where we think of $k$ as
$\ll n$, or even a constant); the majority function $\MAJ(x) =
\sgn(x_1 + \ldots + x_n)$.  The majority function is special instance
of linear threshold functions $f(x) = \sgn(a_0 + a_1x_1 + \ldots +
a_nx_n)$, $a_i \in \R$, also known as weighted-majority functions, or
halfspaces. Another important voting scheme in boolean function
analysis is ${\sf TRIBES}_{w,s}:\bits^{ws}\to\bits$, the $s$-way
$\sfor$ of $w$-way $\sfand$'s of disjoint sets of variables (where we
think of $-1$ as true and $1$ as false). In ${\sf TRIBES}_{w,s}$, the
candidate $-1$ is elected iff at least one member of each of the $s$
disjoint tribes of $w$ members votes for $-1$.\medskip 

The following are a few reasonable properties one may expect of a
voting scheme: 

\bit
\itemsep -.5pt 
\item Monotone: if $x_i\leq y_i$ for all $i\in [n]$ then $f(x)\leq
  f(y)$. 
\item Symmetric: $f(\pi(x)) = f(x)$ for all permutations $\pi\in
  S_n$ and $x\in\bn$. 
\item Transitive-symmetric (weaker than symmetric): for all $i,j \in
  [n]$ there exists a permutation $\pi\in S_n$ such that $\pi(i) = j$
  and $f(x) = f(\pi(x))$ for all $x\in \bn$.  \eit 

  Later in this section (for the proof of Arrow's theorem) we will
  also assume that voters vote independently and uniformly; this is
  known as the impartial culture assumption in social choice theory.

\bdefn[influence] Let $f\isafunc$. We say that variable $i\in [n]$ is
pivotal for $x\in\bn$ if $f(x) \neq f(x^{\oplus i})$, where $x^{\oplus
  i}$ is the string $x$ with its $i$-th bit flipped. The influence of
variable $i$ on $f$, denoted $\Inf_i(f)$, is the fraction of inputs
for which $i$ is pivotal. That is, $\Inf_i(f) := \Pr[f(x) \neq
f(x^{\oplus i})]$.  \edefn

For example, $\Inf_i({\sf DICT}_j)$ is 1 if $i = j$ and 0
otherwise. For the majority function over an odd number $n$ of
variables,  $\Inf_i(\MAJ) = {n-1\choose (n-1)/2}\cdot 2^{-(n-1)}$ since
voter $i$ is pivotal iff the votes are split evenly among the other
$n-1$ voters.  By Stirling's approximation, this quantity is $\sim
\sqrt{2/\pi n} = \Theta(1/\sqrt{n})$.

\bdefn[derivative] Let $f:\bn\to\R$. The $i$-th derivative of $f$ is
the function $(D_if)(x) := \half (f(x^{i\leftarrow 1}) -
f(x^{i\leftarrow -1}))$, where $x^{i\leftarrow b}$ is the string $x$
with its $i$-th bit set to $b$.  \edefn

Note that if $f$ is a boolean function then $(D_if)(x) = \pm 1$ if $i$
is pivotal for $f$ at $x$, and 0 otherwise, and therefore
$\E[(D_if)(x)^2] = \Inf_i(f)$.  We will adopt $\E[(D_if)(x)^2]$ as the
generalized definition of the influence of variable $i$ on $f$ for
real-valued functions $f:\bn\to\R$.

\bthm[Fourier expressions for derivatives and influence]
\label{thm:bderivativefourier}
Let $f:\bn\to\R$. Then 
\ben
\itemsep -.5pt
\item $(D_if)(x) = \sum_{S\ni i} \hatfs\chi_{S\backslash i}(x)$. 
\item $\Inf_i(f) = \sum_{S\ni i} \hatfs^2$. 
\een 
\ethm

\bpf The first identity holds by noting that $(D_i\chis)(x) =
\chi_{S\backslash i}(x)$ if $i\in S$ and $0$ otherwise, along with the
fact that $D_i$ is a linear operator, {\it i.e.}  $D_i(\alpha f + g) =
\alpha (D_if) + (D_ig)$.  The second identity then follows by applying
Parseval to $\Inf_i(f) = \E[(D_if)(x)^2]$. \epf

\bprop[influence of monotone functions]
\label{prop:bmonoinf}
Let $f\isafunc$ be a monotone function. Then $\Inf_i(f) = \hatf(i)$. 
\eprop

\bpf If $f$ is monotone then $(D_if)(x) \in \set{0,1}$ and so
$\Inf_i(f) = \E[(D_if)(x)^2] = \E[(D_if)(x)] = \wh{D_if}(\emptyset) =
\hatf(i)$.  Here the final equality uses the Fourier expansion of
$D_if$ given by Theorem \ref{thm:bderivativefourier}. \epf

An immediate corollary of Proposition \ref{prop:bmonoinf} is that
monotone, transitive-symmetric functions $f$ have $\Inf_i(f) \leq
1/\sqrt{n}$ for all $i\in [n]$. This follows from the fact that
transitive-symmetric functions satisfy $\hat{f}(i) = \hat{f}(j)$ for
all $i,j\in [n]$, along with the bound $\sumi\hat{f}(i)^2 \leq
\sumalls \hatfs^2 = 1$.

\bdefn[total influence]
 Let $f:\bn\to\R$. The total influence of $f$ is $\Inf(f) :=
 \sumi\Inf_i(f)$. 
\edefn

If $f$ is a boolean function, then
\[ \Inf(f) = \sumi\Pr[f(x)\neq f(x^{\oplus i})] = \sumi
\Ex[\indi(f(x)\neq f(x^{\oplus i}))] = \Ex\Big[\lsum_{i=1}^n
\indi(f(x)\neq f(x^{\oplus i}))\Big].\] The quantity $\sumi
\indi(f(x)\neq f(x^{\oplus i}))$ is known as the sensitivity of $f$ at
$x$, and so the total influence of a boolean function is also known as
its average sensitivity.  If $f$ is viewed as a 2-coloring of the
boolean hypercube, the total influence can also be seen to be equal to
$n$ times the fraction of bichromatic edges.\medskip

If $f$ is a monotone boolean function, we see that $\Inf(f) =
\sumi\hatf(i) = \sumi\E[f(x)x_i] = \E[f(x)(x_1+\ldots+x_n)]$. Recall
that if $f$ is boolean then $f(x)(x_1+\ldots+x_n) = 1$ if $f(x) =
\sgn(x_1+\ldots +x_n)$ and $-1$ otherwise. Therefore the total
influence a monotone boolean function, when viewed as a voting scheme,
measures the expected difference between the number of voters whose
vote agrees with the outcome of the election and the number whose vote
disagrees. It is reasonable to expect this quantity to large, and the
next proposition states that (if $n$ is odd) it is maximized by the
majority function. 

\bprop[$\MAJ$ maximizes sum of linear coefficients] 
\label{prop:bmajlin}
Let $n$ be odd.
Among all boolean functions $f\isafunc$, the quantity
$\sumi\hat{f}(i)$ is maximized by $\MAJ(x) =
\sgn(x_1+\ldots+x_n)$. Consequently, if $f$ is monotone then $\Inf(f)
\leq \Inf(\MAJ) \sim \sqrt{2n/\pi}$. 
\eprop

\bpf Note that $\sumi\hat{f}(i) = \E[f(x)(x_1+\ldots+x_n)] \leq
\E[|x_1+\ldots+x_n|]$ since $f$ is $(\pm 1)$-valued, where the
inequality is tight iff $f(x) = \sgn(x_1+\ldots+x_n) = \MAJ(x)$. For
the second claim recall that $\Inf_i(f) = \hat{f}(i)$ if $f$ is
monotone (Proposition \ref{prop:bmonoinf}), and $\Inf_i(\MAJ) \sim
\sqrt{2/\pi n}$ for all $i\in [n]$. \epf

\bprop[Fourier expression for total influence]
Let $f:\bn\to\R$. Then $\Inf_i(f) = \sumalls |S|\cdot \hatfs^2 =
\sum_{k=1}^n k\cdot \bfW^k(f)$. 
\eprop

The proof of this proposition follows immediately from the Fourier
expression for variable influence given by Theorem
\ref{thm:bderivativefourier}. Notice that each Fourier coefficient is
weighted by its cardinality in the sum, and so total influence may
also be viewed as a measure of the ``average degree'' of $f$'s Fourier
expansion.\medskip

Recall that for functions $f:\bn\to\R$ we have $\Var(f) =
\sum_{S\neq\emptyset}\hatfs^2$ (Proposition \ref{prop:bexandvar}), and
comparing this quantity with the Fourier expression for total
influence yields $\Var(f)\leq \Inf(f)$, the Poincar\'e inequality for
the boolean hypercube. The inequality is tight iff $\bfW^1(f) = 1$,
which for boolean functions implies that $f = \pm \mathsf{DICT}_i(f)$
for some $i\in [n]$. If $f$ is a boolean function and $p := \Pr[f(x)
=1]$, it is easy to check that $\Var(p) = 4p(1-p)$, and so the
Poincar\'e inequality can be equivalently stated as $4p(1-p) \leq
\Inf(f) = n\cdot (\text{fraction of bichromatic edges})$.  The
Poincar\'e inequality is therefore an edge-isoperimetric inequality
for the boolean hypercube (where we view boolean functions as
indicators of subsets of $\bn$): for any $p$, it gives a lower bound
on the number of boundary edges between $A$ and $\overline{A}$ where
$A\sse \bn$ has density $p$.  This is a sharp bound when $p = 1/2$,
but not when $p$ is small. For smaller densities we have the bound
$2\alpha\log_2(1/\alpha)\leq \Inf(f)$, where $\alpha := \min\set{\Pr[f(x) =
  1],\Pr[f(x) = -1]}$.  This in turn is sharp whenever $\alpha = 2^k$,
achieved by the {\sf AND} of $k$ coordinates.

\subsection{Noise stability and Arrow's theorem} 
 
Let $\rho\in [0,1]$ and fix $x\in\bn$. Let $N_\rho(x)$ be the
distribution on $\bn$ where $y\sim N_\rho(x)$ if for all $i\in [n]$,
$y_i = x_i$ with probability $\rho$, and $y_i$ is uniformly random $\pm
1$ with probability $1-\rho$.  More generally, for $\rho\in [-1,1]$,
we have that $N_\rho(x)$ is the distribution on strings $y$ where 
\[
y_i = \left\{
\begin{array}{cl}
x_i & \text{with probability $\half + \half\rho$} \\
-x_i & \text{with probability $\half - \half\rho$}.   
\end{array}
\right.
\] 
If $x\sim\bn$ is uniformly random and $y\sim N_\rho(x)$, we say that
$x$ and $y$ are $\rho$-correlated strings; equivalently, $x$ and $y$
are $\rho$-correlated if they are both uniformly random and
$\E[x_iy_i] = \rho$ for all $i\in [n]$.

\bdefn[noise stability]
Let $f:\bn\to\R$ and $\rho \in [-1,1]$. The noise stability of $f$ at
noise rate $\rho$ is 
\[ \Stab_\rho(f) := \E[f(x)f(y)], \quad \text{where $x,y$ are
  $\rho$-correlated strings.}  \] 
\edefn

For example $\Stab_\rho(\pm 1) = 1$, $\Stab_\rho(\mathsf{DICT}_i) =
\rho$, and $\Stab_\rho(\chis) = \rho^{|S|}$.  Tomorrow we will prove
Sheppard's formula: $\lim_{n\to\infty}\Stab_\rho(\MAJ) =
1-\frac{2}{\pi}\arccos(\rho)$. In particular, if $\rho = 1-\delta$ we
have $\Stab_\rho(\MAJ) = \Theta(\sqrt{\delta})$.

\bdefn[noise operator]
Let $\rho\in [-1,1]$. The noise operator $T_\rho$ on functions
$f:\bn\to\R$ acts as follows: $(T_\rho f)(x)  := \E_{y\sim N_\rho(x)}[f(y)]$. 
\edefn 

\bprop[Fourier expressions for noise operator and stability]
\label{prop:bnoise}
Let $\rho\in [-1,1]$ and $f:\bn\to\R$. Then 
\ben
\itemsep -.5pt
\item $(T_\rho f)(x) = \sumalls\rho^{|S|}\hatfs\chis(x)$.  
\item $\Stab_\rho(f) = \sumalls \rho^{|S|}\hatfs^2$. 
\een 
\eprop

\bpf The first identity follows from the linearity of the noise
operator, along with the observation that $(T_\rho\chis)(x) =
\rho^{|S|}\chis(x)$. The second holds by noting that 
\[ \Stab_\rho(f) = \Ex_{(x,y)\,\rho\text{-corr}}[f(x)f(y)] =
\Ex_x[f(x)(T_\rho f)(x)] = \sumalls \hatfs\wh{T_\rho f}(S) = \sumalls
\rho^{|S|}\hatfs^2. \]  \epf

Suppose there is an election with $n$ voters and three candidates: $A,
B$ and $C$. Each voter ranks the candidates by submitting three bits
indicating her preferences: whether the prefers $A$ to $B$ (say, $-1$
if so and $1$ otherwise), and similarly for $B$ versus $C$ and $C$
versus $A$. Clearly a rational voter cannot simultaneously prefer $A$
to $B$, $B$ to $C$ and $C$ to $A$; her ordering of the candidates must
be non-cyclic.
  
\bdefn[rational] A triple $(a,b,c)\in\bits^3$ is rational if not all
three bits are equal (i.e., $(a,b,c)$ defines a total ordering, and is
a valid preference profile). We define the function
$\NAE:\bits^3\rightarrow\{1,0\}$ to be $1$ iff not all three bits are
equal.  \edefn

Now suppose the preferences of the $n$ voters are aggregated into
three $n$-bit strings $x,y$ and $z$, and the aggregate preference of
the electorate is represented by the triple $(f(x),f(y),f(z))$ for
some boolean function $f\isafunc$. Clearly we would like for the the
outcome of the election to be rational; that is, $\NAE(f(x),f(y),f(z))
=1$.

\begin{fact}[Condorcet's paradox \cite{Con85}] 
  With $\MAJ$ as the aggregating function it is possible that all
  voters submit rational preferences and yet the aggregated preference
  string is irrational.
\end{fact}

\begin{figure}[h]
\begin{center}
\begin{tabular}{|c|rrr|c|} 
\hline
& $v_1$ & $v_2$ & $v_3$ & $\MAJ$ \\ \hline
$A > B$\,? & +1 & +1 & $-1$ & +1 \\ 
$B > C$\,? & +1 & $-1$ & +1 & +1 \\
$C > A$\,? & $-1$ & +1 & +1 & +1 \\ \hline 
\end{tabular}
\caption{An instance of Condorcet's paradox} 
\end{center} 
\end{figure}
\vspace{-18pt} 
 
\bthm[Arrow's impossibility theorem \cite{Arr50}] Suppose $f$ is an aggregating
function that always produces a rational outcome if all voters vote
rationally. Then $f = \pm{\sf DICT}_i$ for some $i\in [n]$.  If $f$ is
further restricted to be unanimous ({\it i.e.}  $f(1,\ldots,1) = 1$
and $f(-1,\ldots,-1) = -1$; certainly a very reasonable assumption)
then $f$ must be a dictator.  \ethm

The main result of this section is a robust version of Arrow's
impossibility theory due Gil Kalai \cite{Kal02}.  It expresses the
probability that an aggregating function $f$ produces a rational
outcome in terms of the noise stability of $f$, under the impartial
culture assumption (each voter selects an $\NAE$-triple
$(x_i,y_i,z_i)$ uniformly and independently).

\bthm[Kalai] 
\label{thm:bkalai}
$\E[\NAE(f(x),f(y),f(z))] = \frac{3}{4} - \frac{3}{4}\cdot
\Stab_{-1/3}(f)$, where the expectation is taken with respect to the
impartial culture assumption. \ethm

\bpf 
Using the arithmetization $\NAE(a,b,c) = \frac{3}{4} -
\inv{4}(ab+bc+ac)$, we first note that 
\begin{eqnarray*}
\E[\NAE(f(x),f(y),f(z))] &=& {\lfrac 3 4} - {\lfrac 1
  4}\big(\E[f(x)f(y)] + \E[f(y)f(z)] + \E[f(x)f(z)] \big)  \\
&=& {\lfrac 3 4} - {\lfrac 3 4} \E[f(x)f(y)], 
\end{eqnarray*}
where again, all expectations are taken with respect to the impartial
culture assumption. Since $\E[x_iy_i] = -1/3$ if $(x_i,y_i,z_i)$ is an $\NAE$-triple, the
quantity $\E[f(x)f(y)]$ is exactly $\Stab_{-1/3}(f)$ and the proof is
complete.
\epf 
 
Theorem \ref{thm:bkalai} does indeed imply Arrow's
impossibility theorem since  
\[ {\lfrac 3 4} - {\lfrac 3 4} \cdot \Stab_{-1/3}(f) = {\lfrac 3 4} -
{\lfrac 3 4} \sum_{k=0}^n (-{\lfrac 1 3})^k \cdot \bfW^k(f) \leq
{\lfrac 7 9} + {\lfrac 2 9} \cdot \bfW^1(f), \] and so if
$\E[\NAE(f(x),f(y),f(z))] = 1$ then $\bfW^1(f) \geq 1$.  Furthermore
note that the probability of an irrational outcome is at least
$1-\eps$ then $\bfW^1 \geq 1-O(\eps)$.  By a theorem of E. Friedgut,
G. Kalai and A. Naor \cite{FKN02}, if $\bfW^1(f) \geq 1-\eps$ then
$f$ is $O(\eps)$-close to $\pm {\sf DICT}_i$ for some $i\in
[n]$. Therefore Kalai's theorem is in fact a robust version of Arrow's
impossibility theorem: if most rational voter preference profiles
aggregate to a rational outcome, then the aggregating function must be
close to a dictator or anti-dictator.\medskip

We conclude by giving an upper bound on level-1 Fourier weight of
transitive-symmetric functions.  By Theorem \ref{thm:bkalai}, this
gives an upper bound on the probability that such functions aggregate
rational voter preference profiles to a rational outcome.  We will
also prove a generalization of this fact (Proposition
\ref{prop:b2overpi}) using the Berry-Ess\'een theorem tomorrow.
 
\bprop[$\bfW^1(f)$ of transitive-symmetric functions]
Suppose $\hatf(i) = \hatf(j)$ for all $i,j\in [n]$. Then $\bfW^1(f)
\leq {\lfrac 2 \pi} + o_n(1)$. 
\eprop 

\bpf First note that $\sumi\hatf(i)^2 = n\cdot\hatf(1)^2 = n \cdot
\big(\inv{n}\sumi\hatf(i)\big)^2 =
\inv{n}\big(\sumi\hatf(i)\big)^2$. The claim then follows since
we have seen that $\sumi\hatf(i) \leq \sumi\wh{\MAJ}(i) \sim
\sqrt{2n/\pi}$ (Proposition \ref{prop:bmajlin}).  \epf

\pagebreak
\section{Noise stability and small set expansion} 

\begin{center}
\large{Tuesday, 28th February 2012} 
\end{center}

\bit
\item Puzzle: Compute the Fourier expansion of $\MAJ_n$. Hint:
  consider $T_\rho D_i\MAJ(1,\ldots,1)$. 
\eit 

Roughly speaking, the central limit theorem states that if
$X_1,\ldots,X_n$ are independent random variables where none of the
$X_i$'s are ``too dominant'', then $S = \sumi X_i$ is distributed like
a Gaussian as $n\to\infty$.  As a warm-up, we begin by giving another
proof of the fact that $\Inf(\MAJ) \sim \sqrt{2n/\pi}$, this time
using the central limit theorem. First note that
 
\[ \Inf(\MAJ) = \E\Big[\MAJ(x)\lsum_{i=1}^n x_i\Big] = \E\Big[\Big|\lsum_{i=1}^n
x_i\Big|\Big] = \sqrt{n}\cdot \E\Big[\Big|\lsum_{i=1}^n
\frac{x_i}{\sqrt{n}}\Big|\Big] \] 

Now by the central limit theorem we know that $\inv{\sqrt{n}}\sumi x_i
\to \calg\sim N(0,1)$ as $n\to\infty$, and since $\E[|\calg|] =
\sqrt{2/\pi}$, we conclude that $\Inf(\MAJ) \sim \sqrt{2/\pi}\cdot
\sqrt{n}$.

\bdefn[reasonable r.v.]  Let $B\geq 1$. We say that a random variable
$X$ is $B$-reasonable if $\E[X^4] \leq B\cdot
\E[X^2]^2$. Equivalently, $\norm{X}_4 \leq B^{1/4}\cdot \norm{X}_2$.
\edefn

For example, a uniformly random $\pm 1$ bit ({\it i.e.} a Rademacher
random variable) is $1$-reasonable, and a standard Gaussian is
3-reasonable. The Berry-Ess\'een theorem \cite{Ber41, Ess42} is a
finitary version of the central limit theorem, giving explicit bounds
on the rate at which reasonable random variables converge towards the
Gaussian distribution.  

\bthm[Berry-Ess\'een] Let $X_1,\ldots,X_n$ be independent,
$B$-reasonable random variables satisfying $\E[X_i] =0$. Let
$\sigma_i^2 := \E[X_i^2]$ and suppose $\sumi\sigma_i^2 =1$. Let $S =
X_1 + \ldots + X_n$ and $\calg\sim N(0,1)$. For all $t\in \R$,
\[ |\Pr[S\leq t] - \Pr[\calg \leq t]| \leq
  O(\eps)\]
where $\eps = \big(B\cdot \sumi \sigma_i^4\big)^{1/2} \leq
\sqrt{B}\cdot \max\set{|\sigma_i|}$. 
\ethm

We prove the Berry-Ess\'een theorem with a weaker bound of $\eps = \big(B\cdot
\sumi\sigma_i^4\big)^{1/5}$ in Section \ref{sec:bbe}.

\subsection{Sheppard's formula and $\Stab_\rho(\MAJ)$}
\label{sec:sheppard}
 
\bdefn[$\rho$-correlated Gaussians] 
\label{def:bcorrgauss}
Let $\calg,\calg'\sim N(0,1)$ be
independent standard Gaussians. Set $\calh = (\calg,\calg')\cdot
(\rho,\sqrt{1-\rho^2}) := \rho\cdot \calg + \sqrt{1-\rho^2}\cdot
\calg'$. Then $\calg$ and $\calh$ are $\rho$-correlated
Gaussians. Note that if $\calg$ and $\calh$ are $\rho$-correlated
Gaussians then $\E[\calg\calh] = \rho \cdot \E[\calg^2] +
\sqrt{1-\rho^2}\cdot \E[\calg]\E[\calg'] = \rho$.\edefn

\bthm[\cite{She99}]
Let $\calg$ and $\calh$ be $\rho$-correlated Gaussians. Then
$\Pr[\sgn(\calg)\neq \sgn(\calh)] = \arccos(\rho)/\pi$. 
\ethm 

\bpf First recall that $\sgn(\vec{u}\cdot \vec{v})$ is determined by
which side of the halfspace normal to $\vec{u}$ the vector $\vec{v}$
falls on. Since $\calh = (\rho,\sqrt{1-\rho^2})\cdot (\calg, \calg')$,
and $\calg = (1,0)\cdot (\calg,\calg')$, it follows that
$\Pr[\sgn(\calg)\neq\sgn(\calh)]$ is precisely the probability that
the halfspace normal to $(\calg,\calg')$ splits the vectors
$(\rho,\sqrt{1-\rho^2})$ and $(1,0)$. Therefore, 
\[ \Pr[\sgn(\calg)\neq\sgn(\calh)] = {\lfrac 1 \pi} \cdot (\text{angle between
    $(\rho,\sqrt{1-\rho^2})$ and $(1,0)$}) =
{\lfrac 1 \pi} \arccos(\rho).\] \epf

Next, we use Sheppard's formula to prove that $\Stab_\rho(\MAJ)\to 1-
\frac{2}{\pi}\arccos(\rho)$ as $n\to\infty$.  First recall that
\[ \Stab_\rho(\MAJ) = \Ex_{(x,y)\,\rho\text{-corr}}[\MAJ(x)\MAJ(y)] =
1-2\Pr[\MAJ(x)\neq \MAJ(y)],\] and so it suffices to argue that
$\Pr[\MAJ(x)\neq \MAJ(y)]\to {\lfrac 1 \pi}\arccos(\rho)$. Next, we view
\[\MAJ(x) = \sgn\Big(\frac{x_1+\ldots+x_n}{\sqrt{n}}\Big),\quad
\MAJ(y) = \sgn\Big(\frac{x_1 + \ldots + x_n}{\sqrt{n}}\Big), \]
and note that  
\[
\E\bigg[\Big(\lsum_{i=1}^n\frac{x_i}{\sqrt{n}}\Big)
\cdot\Big(\lsum_{i=1}^n\frac{y_i}{\sqrt{n}}\Big)\bigg] =
\E\bigg[{\lfrac 1 n}\lsum_{i=1}^n x_i\cdot \lsum_{i=1}^n y_i\bigg] =
{\lfrac 1 n}\sumi \E[x_i y_i] = \rho. \] While the standard central limit
theorem tells us that $\vec{X} = (x_1+\ldots+x_n)/\sqrt{n}$ and
$\vec{Y} = (y_1+\ldots + y_n)/\sqrt{n}$ each individually converges
towards the standard Gaussian $\calg\sim N(0,1)$, the two-dimensional
central limit theorem states that $(\vec{X},\vec{Y})$ actually
converge to $\rho$-correlated Gaussians $(\calg,\calh)$ as $n \to
\infty$.  In fact, the two-dimensional Berry-Ess\'een theorem
quantifies this rate of convergence, bounding the error by $\pm
O(1/\sqrt{n})$ as long as $\rho$ is bounded away from $\pm
1$. Combining this with Sheppard's formula, we conclude that
\[ \Pr[\MAJ(x)\neq \MAJ(y)] \to\Pr[\sgn(\calg)\neq\sgn(\calh)] =
{\lfrac 1 \pi} \arccos(\rho) \quad \text{as $n\to\infty$}. \] 
Since 
\[ \Stab_\rho(\MAJ)\to 1 - {\lfrac 2 \pi}\arccos(\rho) = {\lfrac 2 \pi}
+ {\lfrac 3 \pi}\rho^3 + \ldots + \tbinom{k-1}{(k-1)/2} {\lfrac
  4 {\pi k 2^k}}\cdot \rho^k + \ldots\]
and $\Stab_\rho(\MAJ) = \sum_{k=0}^n \rho^k\cdot \bfW^k(\MAJ)$, we have
that $\bfW^k(\MAJ)\to \tbinom{k-1}{(k-1)/2} {\lfrac
  4 {\pi k 2^k}}$ as $n\to\infty$. 

\subsection{The noisy hypercube graph}

Let $\rho\in [-1,1]$. The $\rho$-noisy hypercube graph is a complete
weighted graph on the vertex set $\bn$, where the weight on an edge
$(x,y)$ is the probability of getting $x$ and $y$ when drawing
$\rho$-correlated strings from $\bn$.  Equivalently, $\wt(x,y) =
\Pr[x\leftarrow \calu]\Pr[y\leftarrow N_\rho(x)] = 2^{-n}\cdot
(\half-\half\rho)^{\Delta(x,y)}(\half + \half \rho)^{n-\Delta(x,y)}$,
and the sum of weights of all edges incident to any $x\in\bn$ is
exactly $2^{-n}$. \medskip

Recall that if $f$ is a boolean function, then $\Stab_\rho(f) =
\E_{(x,y)\,\rho\text{-corr}}[f(x)f(y)] = 1-2\Pr[f(x)\neq f(y)]$, or
equivalently, $\Prx_{(x,y)\,\rho\text{-corr}}[f(x)\neq f(y)] = {\lfrac
  1 2} - {\lfrac 1 2}\cdot \Stab_\rho(f)$.  Viewing $f$ as the
indicator of a subset $A_f\sse \bn$, the quantity $\Pr[f(x)\neq f(y)]$
is the sum of weights of edges going from $A_f$ to its complement
$\overline{A_f}$. Therefore, roughly speaking, a function $f$ is noise
stable iff the sum of weights of edges contained within $A_f$ and
$\overline{A_f}$ is large. On Friday we will prove the Majority Is
Stablest theorem due to E. Mossel, R. O'Donnell, and K. Olezkiewicz
\cite{MOO10}: \smallskip

\begin{quote} {\bf Majority Is Stablest.} Fix a constant $0 < \rho < 1$,
  and let $f\isafunc$ be a balanced function.  It is easy to see that
  $\Stab_\rho(f)$ is maximized when $\bfW^1(f) = 1$, in which case $f
  = \pm {\sf DICT}_i$ and $\Stab_\rho({\sf DICT}_i) = \rho$. However,
  if we additionally require that $\Inf_i(f) \leq \tau$ for all $i\in
  [n]$, then the theorem states that $\Stab_\rho(f) \leq
  \Stab_\rho(\MAJ) + o_\tau(1)$; {\it i.e.} the set $A$ of density
  $1/2$ with all influences small that maximizes sum of weights of
  edges within $A$ and $\overline{A}$ is $A_\MAJ$.
\end{quote}\smallskip

Let $A\sse\bn$ and $\indi_A(x):\bn\to\set{0,1}$ be its indicator
function, and let $\alpha := \E[\indi_A(x)] = |A|\cdot 2^{-n}$ denote
its density.  Recall that $\Stab_\rho(\indi_A) =
\E[\indi_A(x)\indi_A(y)] = \Pr[x\in A\ \&\ y\in A]$, or equivalently,
$\Pr[y\in A\st x\in A] = \inv{\alpha}\cdot \Stab_\rho(\indi_A)$, and
so $\Stab_\rho(\indi_A)$ is the probability that a random walk (where
each coordinate is independently flipped with probability
$(\half-\half\rho)$) starting at a point $x\in A$ remains in $A$,
normalized by the density of $A$.  Tomorrow we will prove (a specific
instance of) the small set expansion theorem:\smallskip

\begin{quote} {\bf Small Set Expansion.} $\Stab_\rho(\indi_A) \leq
  \alpha^{2/(1+\rho)}$, or equivalently, $\Pr[y\in A\st x\in A] \leq
  \alpha^{(1-\rho)/(1+\rho)}$.  In particular, if $\alpha$ is small,
  the probability that a random walk starting in $A$ landing outside
  $A$ is very high.
\end{quote} \smallskip

Recall that $\Stab_\rho(\indi_A) = \bfW^0(\indi_A) + \rho\cdot
\bfW^1(\indi_A) + \rho^2\cdot \bfW^2(\indi_A) + \ldots$, and so
$\frac{d}{d\rho}\Stab_\rho(\indi_A)\big|_{\rho = 0} =
\bfW^1(\indi_A)$. As a direct corollary of the small set expansion
theorem we have the following bound on $\bfW^1(\indi_A)$: 
\[ \bfW^1(\indi_A) = {\lfrac d {d\rho}}\Stab_\rho(\indi_A)\Big|_{\rho
  = 0} \leq {\lfrac d {d\rho}} \alpha^{2/(1+\rho)}\Big|_{\rho = 0} =
2\alpha^2\ln(1/\alpha). \] We now give a self-contained proof this
fact, due to Talagrand \cite{Tal96}:
  
\bthm[level-1 inequality]
\label{thm:btallem}
Let $f:\bn\to\zbits$ and $\alpha=\E[f]$. Then $\mathcal{W}_1(f) = O(\alpha^2
\ln(1/\alpha))$. 
\ethm

\bpf First consider an arbitrary linear form $\ell(x) = \sumi a_ix_i$
normalized to satisfy $\sum a_i^2 = 1$.  For any $t_0\geq 1$, we
partition $x \in \bn$ according to whether $|\ell(x)| < t_0$ or
$|\ell(x)|\geq t_0$, and note that
\[ \E[f(x) \ell(x)] = \E\big[\indi(|\ell(x)|<t_0)\cdot f(x) \ell(x)\big] +
\E\big[\indi(|\ell(x)|\geq t_0)\cdot f(x) \ell(x)\big]. \] The first summand
is at most $\alpha\cdot t_0$, and the second is at
most
\[ \int_{t_0}^\infty 2e^{-t^2/2} \ dt \leq 2 \int_{t_0}^\infty
te^{-t^2/2} \ dt = \left[-2e^{-t^2/2}\right]_{t_0}^\infty = 2\cdot
e^{-t_0^2/2}\] by Hoeffding, where the inequality holds since $t_0\geq
1$. Choosing $t_0 = (2\ln(1/\alpha))^{1/2}\geq 1$, we get
\begin{equation}
 \E[f(x) \ell(x)] \leq O(\alpha
 \sqrt{\ln(1/\alpha)}). \label{eq:btallem}
\end{equation}
Now let $\ell(x) = \inv{\sigma}\sumi \hatf(i)x_i$ where $\sigma =
\sqrt{\bfW_1(f)}$ (if $\sigma=0$ we are done). Note that 
\[ \E[f(x) \ell(x)] = \sum_{i=1}^n \hatf(i)\hat{\ell}(i) = 
{\lfrac 1 \sigma} \sumi\hatf(i)^2 = \sqrt{\bfW_1(f)}.\] 

The claimed inequality then follows by applying (\ref{eq:btallem}). 
\epf

\bprop[$\frac{2}{\pi}$ theorem]
\label{prop:b2overpi}
Let $f\isafunc$ with $|\hatf(i)| \leq \eps$ for all $i\in [n]$. Then
$\bfW^1(f) \leq \frac{2}{\pi} + O(\eps)$. 
\eprop 

\bpf Let $\sigma = \sqrt{\bfW^1(f)}$ and assume without loss of
generality that $\sigma\geq 1/2$ (otherwise $\bfW^1(f) < 1/4 <
\frac{2}{\pi}$ and we are done).  Let $\ell(x) = \inv{\sigma}\sumi
\hatf(i)x_i$ where $|\hat{\ell}(i)| \leq 2\eps$ for all $i\in [n]$.
Note that $\E[f(x)\ell(x)] = \sigma$ and $\E[\ell(x)f(x)] \leq
\E[|\ell(x)|]$. Applying Berry-Ess\'een gives us $\E[|\ell(x)|]
\approx_{O(\eps)} \E[|\calg|] = \sqrt{2/\pi}$ and this completes the
proof (technically Berry-Ess\'een only yields a bound on closeness in
cdf-distance, but this can be translated into a bound on closeness in
first moments). \epf

A dual to Proposition \ref{prop:b2overpi} holds as well: if
$\bfW^1(f) \geq \frac{2}{\pi}-\eps$, then $f$ is
$O(\sqrt{\eps})$-close to a linear threshold function (in fact the LTF
is simply $\sgn\big(\sumi \hatf(i)x_i\big)$).  This is the crux of the
result that the class of linear threshold functions is testable with
$\poly(1/\eps)$-queries \cite{MORS10}.

\subsection{Bonami's lemma} 
 
The next theorem, due to Bonami \cite{Bon70}, states that low degree
multilinear polynomials of Rademachers are reasonable random variables
(this is sometimes known as $(4,2)$-hypercontractivity).

\bthm[Bonami]
\label{thm:bbonami} 
Let $f:\bn\rightarrow\R$ be a multilinear polynomial of degree at most
$d$. Then $\|f\|_4\leq \sqrt{3}^d \|f\|_2$. Equivalently, $\E[f(x)^4]
\leq 9^d \cdot \E[f(x)^2]^2$ \ethm 

\bpf We proceed by induction on $n$.  If $n=0$ then $f$ is the
constant and the inequality holds trivially for all $d$.  For the
inductive step, let
\[ f(x_1,\ldots,x_n) = g(x_1,\ldots,x_{n-1}) +
x_nh(x_1,\ldots,x_{n-1}) \] Notice that $g$ has degree at most $d$,
$h$ has degree at most $d-1$, and both are polynomials in $n-1$
variables. Notice also that the random variable $x_n$ is independent
of both $g$ and $h$. Therefore, we have: 
\begin{eqnarray*}
\E[f^4] &=& \E[(g+x_nh)^4] \\
&=& \E[g^4] + 3\E[x_n]\E[g^3h] + 6\E[x_n^2]\E[g^2h^2] +
3\E[x_n^3]\E[gh^3] + \E[x_n^4]\E[h^4] \end{eqnarray*} where we used
independence for the second equality. Now note that
$\E[x_n]=\E[x_n^3]=0$, $\E[x_n^2]=\E[x_n^4]=1$, and $\E[g^2h^2]
\leq \sqrt{\E[g^4]}\sqrt{\E[h^4]}$ by Cauchy-Schwarz. Therefore, 
\begin{eqnarray*}
\E[f^4] &\leq& \E[g^4] + 6\sqrt{\E[g^4]}\sqrt{\E[h^4]} + \E[h^4] \\
&\leq & 9^d \E[g^2]^2 + 6\sqrt{9^d\E[g^2]^2}\sqrt{9^{d-1}\E[h^2]^2} +
9^{d-1}\E[h^2]^2 \\
&=& 9^d\cdot (\E[g^2]^2 + 2\E[g^2]\E[h^2] + \lfrac{1}{9} \E[h^2]^2) \\
&\leq & 9^d\cdot  (\E[g^2]+\E[h^2])^2 
\end{eqnarray*}
To complete the proof, notice that: 
\begin{eqnarray*}
\E[f^2] &=& \E[(g+x_n h)^2] \\
&=& \E[g^2] + 2\E[x_n]\E[gh] + \E[x_n^2]\E[h^2] \\
&=& \E[g^2] + \E[h^2] 
\end{eqnarray*}
and so we have shown that $\E[f^4]\leq 9^d \cdot \E[f^2]^2$. 
\epf 
 
\pagebreak 
 
\section{KKL and quasirandomness}

\begin{center}
\large{Wednesday, 29th February 2012} 
\end{center}

\bit
\item Open Problem: Prove that among all functions $f\isafunc$ with
  $\deg(f)\leq d$, the quantity $\sumi\hat{f}(i)$ is maximized by
  $\MAJ_d$. Less ambitiously, show $\sumi \hatf(i)= O(\sqrt{\deg(f)})$. 
\eit

\subsection{Small set expansion} 
 
We begin by proving the $\rho = 1/3$ case of the small set expansion
theorem: let $A\sse \bn$ be a set of density $\alpha = |A|\cdot 2^{-n}$, and
$\indi_A:\bn\to\zbits$ be its indicator function. We will need the
following variant of Bonami's lemma; its proof is identical to that of
Theorem \ref{thm:bbonami}.

\bthm[Bonami'] 
\label{thm:bbonamip}
Let $f:\bn\to\R$. Then $\norm{T_{1/\sqrt{3}}f}_4 \leq
\norm{f}_2$. \ethm 

Theorem \ref{thm:bbonamip} is a special case of the hypercontractivity
inequality \cite{Bon70, Gro75, Bec75}: if $1\leq p \leq q \leq \infty$
and $\rho \leq \sqrt{\frac{p-1}{q-1}}$, then $\norm{T_\rho f}_q \leq
\norm{f}_p$.

\bthm[SSE for $\rho = 1/3$]
\label{thm:bsse}
Let $A\sse\bn$. Then
$\Stab_{1/3}(\indi_A) \leq \alpha^{3/2}$, where $\alpha$ is the
density of $A$. 
\ethm

\bpf We will need a corollary of Theorem \ref{thm:bbonamip} that will
also be useful for us when proving the KKL theorem in the next
section: $\norm{T_{1\sqrt{3}}f}_2^2 \leq \norm{f}_{4/3}^2$.  To see
that this holds, we check that
\begin{eqnarray}
\norm{T_{1/\sqrt{3}}f}_2^2 &=& \E[f(x) (T_{1/3}f)(x)] \nonumber \\ & \leq  &
\norm{f}_{4/3}\cdot \norm{T_{1/3}f}_4 \label{eq:bsse1} \\
& = & \norm{f}_{4/3}\cdot \norm{T_{1/\sqrt{3}}\,(T_{1/\sqrt{3}}f)}_4
\nonumber \\ 
&\leq & \norm{f}_{4/3}\cdot \norm{T_{1/\sqrt{3}}f}_2. \label{eq:bsse2} 
\end{eqnarray}
Here (\ref{eq:bsse1}) is by H\"older's inequality and (\ref{eq:bsse2})
by applying Theorem \ref{thm:bbonamip} to $T_{1/\sqrt{3}}f$; dividing
both sides by $\norm{T_{1/\sqrt{3}}f}_2$ yields the claim. Applying
this corollary to $f = \indi_A$ completes the proof:
\[ \Stab_{1/3}(\indi_A) = \sumalls \big({\lfrac 1
  3}\big)^{|S|}\hatfs^2 = \big\|T_{1/\sqrt{3}}\indi_A\big\|_2^2 \leq
\E\big[\indi_A(x)^{4/3}\big]^{3/2} = \alpha^{3/2}.  \] 
Here the second equality is an application of Parseval's, and the
final uses the fact that $\indi_A$ is $\set{0,1}$-valued. 
\epf

\subsection{Kahn-Kalai-Linial}
 
\bthm[\cite{KKL98}]
Let $f\isafunc$ with $\E[f] = 0$, and set $\alpha := \max_{i\in
  [n]}\set{\Inf_i(f)}$. Then $\Inf(f) = \Omega(\log(1/\alpha))$. 
\ethm

\bpf Recall that variable $i$ is pivotal for $x$ in $f$ iff $(D_if)(x)
= \pm 1$, and so $\E[|D_if|] = \Inf_i(f) \leq \alpha$; the plan is to
apply the small set expansion theorem to $D_if$ and sum over all $i\in
[n]$.  We have shown in the proof of Theorem \ref{thm:bsse} that
$\norm{T_{1/\sqrt{3}}g}_2^2 \leq \norm{g}_{4/3}^2$, and applying it to
$g = D_if$ gives us
\begin{equation} \Stab_{1/3}(D_if) \leq \norm{D_if}_{4/3}^2 =
  \E\big[|D_if|^{4/3}\big]^{3/2} =
  \Inf_i(f)^{3/2}. \label{eq:bsse3} \end{equation} We sum both sides
of the inequality over $i\in [n]$, starting with the left hand
side. Recall that $\Stab_{1/3}(D_if) = \sumalls \big({\lfrac 1
  3}\big)^{|S|} \wh{D_if}(S)^2 = \sum_{S\ni i}\big({\lfrac 1
  3}\big)^{|S|-1}\hatfs^2$, and so:
\begin{eqnarray}
\sumi\Stab_{1/3}(D_if)&=& \ \sum_{i=1}^n \sum_{S\ni i} \big({\lfrac 1
  3}\big)^{|S|-1}\hatfs^2 \nonumber \\
&=& \sum_{|S|\geq 1} |S|\,  \big({\lfrac 1 3}\big)^{|S|-1}\hatfs^2
\nonumber \\
&\geq & \sum_{1\leq |S| \leq 2\,\Inf(f)} |S|\,  \big({\lfrac 1
  3}\big)^{|S|-1}\hatfs^2 \nonumber \\
&\geq& \sum_{1\leq |S| \leq 2\,\Inf(f)} 2\,\Inf(f)\cdot  \big({\lfrac 1
  3}\big)^{2\,\Inf(f)-1}\hatfs^2 \label{eq:bkkl1} \\
&\geq& 3\,\Inf(f)\cdot \big({\lfrac 1 9}\big)^{\Inf(f)}.\label{eq:bkkl2} 
\end{eqnarray} 
Here (\ref{eq:bkkl1}) uses the fact that $x\cdot 3^{-(x-1)}$ is a decreasing
function when $x\geq 1$, and (\ref{eq:bkkl2}) uses Markov's inequality
$\sum_{|S|\geq 2\,\Inf(f)} \hatfs^2 \leq \half$ along with the
assumption that $f$ is balanced. Now summing the right hand side of
(\ref{eq:bsse3}) gives us $\sumi\Inf_i(f)^{3/2} \leq \alpha\cdot
\Inf(f)$. Combining both inequalities yields $3\cdot
\big(\inv{9}\big)^{\Inf(f)}\leq \alpha^{1/2}$, and therefore 
$\Inf(f) = \Omega(\log(1/\alpha))$ as claimed.  \epf

\bcor
\label{cor:bkklcor}
Let $f\isafunc$ with $\E[f] = 0$. Then $\max_{i\in [n]}\set{\Inf_i(f)}
= \Omega\big(\frac{\log(n)}{n}\big)$. 
\ecor
 
This bound on the maximum influence is tight for the Ben-Or Linial
{\sf TRIBES} function \cite{BL89}: the $2^k$-way {\sf OR} of $k$-way
{\sf AND}'s of disjoint sets of variables (so $n = k\cdot 2^k$). We
remark that while Corollary \ref{cor:bkklcor} only gives a $\log(n)$
improvement over the $1/n$ bound that follows directly from the
Poincar\'e inequality, this factor makes a crucial difference in many
applications ({\it e.g.} it is the crux of Khot and Vishnoi's
\cite{KV05} counter-example to the Goemans-Linial conjecture
\cite{Goe97, Lin02}).
 
\bcor Let $f\isafunc$ be a balanced, monotone function viewed as a
voting scheme. Both candidates can bias the outcome of the election in
their favor to 99\% probability by bribing a $O\big(\inv{\log(n)}\big)$ fraction of
voters.  \ecor

\subsection{Dictator versus Quasirandom tests}

\bdefn[noisy influence] Let $f:\bn\to\R$ and $\rho\in [0,1]$. The
$i$-th $\rho$-noisy influence of $f$ is $\Inf_i^{(\rho)}(f) :=
\Stab_\rho(D_if) = \sum_{S\ni i}\rho^{|S|-1}\hatfs^2$. 
\edefn

Note that $\Inf_i^{(1)}(f) = \Inf_i(f)$ and $\Inf_i^{(0)}(f) =
\hatf(i)^2$, and intermediate values of $\rho\in (0,1)$ interpolates
between these two extremes -- the larger the value of $\rho$ is, the
more the weight $\hatfs^2$ on larger sets $S$ is dampened by the
attenuating factor $\rho^{|S|-1}$.

\bdefn[quasirandom] Let $f:\bn\to\R$ and $\eps,\delta\in [0,1]$.  We
say that $f$ is $(\eps,\delta)$-quasirandom, or $f$ has
$(\eps,\delta)$-small noisy influences, if $\Inf_i^{(1-\delta)}(f)
\leq \eps$ for all $i \in [n]$.  \edefn

A few prototypical quasirandom functions are the constants $\pm 1$
(these are $(0,0)$-quasirandom), the majority function
(($O(\inv{\sqrt{n}}),0$)-quasirandom), and large parities $\chi_S$
($(1-\delta)^{|S|-1},0)$-quasirandom).  Unbiased juntas, and dictators
in particular, are prototypical examples of functions far from
quasirandom.  The next proposition states that even functions far from
being quasirandom can only have a small number of variables with large
noisy influence:

\bprop
\label{prop:bquasijuntas}
Let $f:\bn\to\R$ with $\Var(f) \leq 1$, and let $J = \{i \in
  [n]\stc \Inf_i^{(1-\delta)}(f) \geq \eps\}$ be the set of coordinates
with large noisy influences. Then $|J| \leq 1/\eps\delta$. 
\eprop

\bpf 
We first note that 
\[ \eps\cdot |J| \leq \sum_{i\in J}\Inf_i^{(1-\delta)}(f) \leq \sumi
\Inf_i^{(1-\delta)}(f) = \sum_{|S|\geq 1}|S|\cdot
(1-\delta)^{|S|-1}\hatfs^2. \] It remains to check that $|S|\cdot
(1-\delta)^{|S|-1} \leq 1/\delta$ for any $S\sse [n]$: to see this
holds, note that $(1-\delta)^{|S|-1}\leq (1-\delta)^{i-1}$ for any
$i\leq |S|$, and so summing over $i$ from $1$ to $|S|$ gives us
$|S|\cdot (1-\delta)^{|S|-1} \leq \sum_{i=1}^{|S|}(1-\delta)^{i-1}
\leq \sum_{i=1}^\infty(1-\delta)^{i-1} = 1/\delta$. We have shown that
$\eps\cdot |J| \leq (1/\delta)\cdot \Var(f) \leq 1/\delta$, and the
proof is complete. \epf

Consider the problem of testing dictators: given blackbox access to a
boolean function $f$, if $f$ is a dictator the test accepts with
probability 1, and if $f$ is $\eps$-far from any of the $n$ dictators
it accepts with probability $1-\Omega(\eps)$.  Implicit in Kalai's
proof of Arrow's impossibility theorem (Theorem \ref{thm:bkalai}) is a
3-query test that comes close to achieving this:
\begin{figure}[h]
\bbox
\begin{center} 3-query $\NAE$ test\vspace{-5pt} 
\end{center}
\ben
\itemsep -.5pt
\item For each $i\in [n]$, pick $(x_i,y_i,z_i)$ uniformly from
  the 6 possible $\NAE$ triples. 
\item Query $f$ on $x$, $y$, and $z$. 
\item Accept iff $\NAE(f(x),f(y),f(z)) = 1$.  
\een 
\ebox
\vspace{-10pt}
\end{figure}
 
Recall that the $\NAE$ test accepts $f$ with probability $\frac{3}{4}
- \frac{3}{4}\cdot \Stab_{-1/3}(f)$, and so if $f = \mathsf{DICT}_i$
for some $i\in [n]$ (in particular, $\Stab_{-1/3}(f) = -1/3$) the test
accepts with probability 1 ({\it i.e.} we have perfect completeness).
However, the $\NAE$ test does not quite satisfy the soundness
criterion. We saw that if the test accepts $f$ with probability
$1-\eps$ then $\bfW^1(f) \geq 1-O(\eps)$, and by a theorem of
Friedgut, Kalai and Naor, $f$ has to $O(\eps)$-close to a dictator or
an anti-dictator; the soundness criterion requires $f$ to be
$O(\eps)$-close to a dictator.  It therefore remains to rule out
functions that are close to anti-dictators, and we will do this using
the Blum-Luby-Rubinfeld linearity test.  Recall that the $\BLR$ test
is a 3-query test that accepts $f$ with probability 1 if $f = \chis$
for some $S\sse [n]$, and with probability $1-\Omega(\eps)$ if $f$ is
$\eps$-far from all the parity functions.  Combining the $\BLR$ and
$\NAE$ tests, we have a 6-query test for dictatorship with perfect
completeness and soundness $1-\Omega(\eps)$.  In fact it is easy to
show that we can perform just one of the two tests, each with
probability $\half$, and reduce the query complexity to 3 while only
incurring a constant factor in the rejection probability.\medskip

As we will see on Saturday, for applications to hardness of
approximation ({\sf UGC}-hardness in particular) it
suffices to design a test that distinguishes dictators from
$(\eps,\eps)$-quasirandom functions, instead of one that distinguishes
dictators from functions $\eps$-far from dictators.

\bdefn[{\sf DICT} vs {\sf QRAND}] Let $0\leq s < c \leq 1$. A $(c,s)$
dictator versus quasirandom test is defined as follows.  Given
blackbox access to a function $f\isafunc$, 
\ben
\itemsep -.5pt
\item The test makes $O(1)$ non-adaptive queries to $f$.
\item If $f$ is a dictator, it accepts with probability at
least $c$. 
\item If $f$ is $(\eps,\eps)$-quasirandom, it accepts with
probability at most $s + o_\eps(1)$.  
\een \edefn

As we will see, often we will need to assume that $f$ is odd ({\it
  i.e.} $f(-x) = -f(x)$ for all $x\in\bn$, or equivalently, $\hatfs =
0$ for all even $|S|$).  Let us consider the $\NAE$ test as a dictator
versus quasirandom test, under the promise that $f$ is odd.  We have
seen that the test has perfect completeness ({\it i.e.} $c = 1$), and
now we determine the value of $s$.  First note that since $f$ is odd,
\begin{eqnarray*} \Pr[\NAE\ \accepts\ f] &=& {\lfrac 3 4} - {\lfrac 3 4}\cdot
  \big(\bfW^0(f) - {\lfrac 1 3}\bfW^1(f) + {\lfrac 1 9}\bfW^2(f) - {\lfrac
    1 {27}}\bfW^3(f) + \ldots \big)\\ 
&=& {\lfrac 3 4} - {\lfrac 3 4}\cdot
  \big(-{\lfrac 1 3}\bfW^1(f)  - {\lfrac
    1 {27}}\bfW^3(f) - {\lfrac 1 {243}\bfW^5(f)} - \ldots\big)\\ 
&=& {\lfrac 3 4} + {\lfrac 3 4}\cdot
  \big({\lfrac 1 3}\bfW^1(f)  + {\lfrac
    1 {27}}\bfW^3(f) + {\lfrac 1 {243}\bfW^5(f)} + \ldots \big)\\ 
&=& {\lfrac 3 4} + {\lfrac 3 4}\cdot \Stab_{1/3}(f). 
\end{eqnarray*} 
Now let $f$ be a $(\eps,\eps)$-quasirandom function.  Applying the
Majority Is Stablest theorem (we will need a statement of it for
functions with $\eps$-small $\eps$-noisy influences instead of
$\eps$-small regular influences), we have
\begin{eqnarray} \Pr[\NAE\ \accepts\ f] &=& {\lfrac 3 4} + {\lfrac 3 4}\cdot
\Stab_{1/3}(f)\nonumber \\ &\leq& {\lfrac 3 4} + {\lfrac 3 4}\cdot
\Stab_{1/3}(\MAJ) + o_\eps(1)  \nonumber \\
&\to& {\lfrac 3 4} + {\lfrac 3 4}\cdot
\big(1-{\lfrac 2 \pi}\arccos({\lfrac 1 3})\big) +
o_\eps(1)\label{eq:bnae1} \\
&=& 0.91226\ldots + o_\eps(1),\nonumber 
\end{eqnarray} 
where (\ref{eq:bnae1}) uses the estimate we proved for
$\Stab_\rho(\MAJ)$ using Sheppard's formula and the central limit
theorem in Section \ref{sec:sheppard}. Therefore we have shown that
the $\NAE$ test is a $(1,0.91226\ldots)$ dictator versus quasirandom
test.\medskip 

We consider two more examples of dictator versus quasirandom tests and
compute their $c$ and $s$ values: the $\rho$-noise test of S. Khot,
G. Kindler, E. Mossel and R. O'Donnell \cite{KKMO07}, and J. H\r{a}stad's ${\sf
  3XOR}_\delta$ test \cite{Has01}.

\begin{figure}[h]
\bbox
\begin{center} KKMO 2-query $\rho$-noise test
\vspace{-5pt} 
\end{center}
\ben
\itemsep -.5pt
\item Let $\rho\in [0,1]$.  Pick $x\in\bn$ uniformly, and $y\sim
  N_\rho(x)$. 
\item Query $f$ on $x$ and $y$. 
\item Accept iff $f(x) = f(y)$. 
\een 
\ebox
\vspace{-10pt}
\end{figure}

First note the $\rho$-noise test accepts $f = \mathsf{DICT}_i$ with
probability $\half + \half\rho$, the probability that $x_i$ is not
flipped in $y$. For soundness, let $f$ be an odd
$(\eps,\eps)$-quasirandom function and note that
\begin{eqnarray*}
  \Pr[{\sf KKMO}\ \accepts\ f] &=& \E\big[{\lfrac 1 2} + {\lfrac 1
    2}f(x)f(y)\big] \\ &=&  {\lfrac 1 2} + {\lfrac 1 2}\cdot \Stab_\rho(f) \\
  &\leq & {\lfrac 1 2} + {\lfrac 1 2} \cdot \Stab_\rho(\MAJ) + o_\eps(1)
  \\
  &\to& {\lfrac 1 2} + {\lfrac 1 2} \cdot \big(1-{\lfrac 2
    \pi}\arccos(\rho)\big)  + o_\eps(1),
\end{eqnarray*} 
where once again we have used the Majority Is Stablest theorem along
with Sheppard's formula (the assumption that $f$ is odd is used in the
application of the Majority Is Stablest theorem, which requires need
$\E[f] = 0$).  Different values of $\rho$ result in different $c$
versus $s$ ratios; for example, if $\rho = 1/\sqrt{2}$ then $c = 0.85$
and $s = 0.75$.  \medskip

For H\r{a}stad's test it will be convenient for us to view our
functions as $f:\F_2^n\to\bits$.  Like in the analyses of the $\NAE$
and $\rho$-noise tests, we will have to assume that $f$ is odd.
\begin{figure}[h]
\bbox
\begin{center} H\r{a}stad's 3-query $\mathsf{3XOR}_\delta$ test
\vspace{-5pt} 
\end{center}
\ben
\itemsep -.5pt
\item Let $\delta \in [0,1]$. Pick $x,y\in\F_2^n$ uniformly and
  independently, and set $z = x+y$. 
\item Pick $x'\sim N_{1-\delta}(x)$. 
\item Query $f$ on $x'$, $y$, and $z$. 
\item Accept iff $f(x')f(y)f(z) = 1$. 
\een 
\ebox
\vspace{-10pt}
\end{figure}

Note that this is identical to the $\BLR$ linearity test, except with
the noisy $x'$ instead of $x$. Once again it is easy to see that
dictators pass with probability $\half + \half(1-\delta) =
1-\frac{\delta}{2}$, and it remains to analyze soundness:
\begin{eqnarray*}
\Pr[\mathsf{3XOR}_\delta\ \accepts\ f] &=&
\Ex_{x,y,x'}\big[\lhalf + \lhalf f(x')f(y)f(z)\big] \\
&=& \lhalf + \lhalf \Ex_{x,y}\Big[\Ex_{x'}[f(x')]f(y)f(z)\Big] \\
&=& \lhalf + \lhalf \Ex_{x,y}\big[(T_{1-\delta}f)(x)f(y)f(z)\big] \\
&=& \lhalf + \lhalf
\Ex_{x}\Big[(T_{1-\delta}f)(x)\Ex_y[f(y)f(x+y)]\Big] \\
&=&  \lhalf + \lhalf \Ex_{x}\big[(T_{1-\delta}f)(x)\wh{f*f}(x)\big] \\
&=&  \lhalf + \lhalf \sumalls \wh{T_{1-\delta}f}(S)\wh{f*f}(S) \ = \
\lhalf + \lhalf \sumalls (1-\delta)^{|S|}\hatfs^3. 
\end{eqnarray*} 
Note that $\sumalls (1-\delta)^{|S|}\hatfs^3 \leq \max_{S\sse
  [n]}\{(1-\delta)^{|S|} |\hatfs|\}$ by Parseval's. Next we claim that
for all $\eps < \delta$, if $f$ is $(\eps,\eps)$-quasirandom then
$(1-\delta)^{|S|} |\hatfs| \leq \sqrt{\eps}$ for all $S$ with odd
cardinality (in particular, $S\neq\emptyset$). To see this, assume for
the sake of contradiction that there exist an $S$ of odd cardinality
for which this inequality does not hold. Then $\sqrt{\eps} <
(1-\delta)^{|S|} |\hatfs| \leq (1-\eps)^{|S|} |\hatfs|$, and
squaring both sides gives us $\eps < (1-\eps)^{2|S|} \hatfs^2 \leq
(1-\eps)^{|S|-1}\hatfs^2 \leq \Inf_i^{(1-\eps)}(f)$ for all $i\in
S$. Since $S\neq\emptyset$, this contradicts the assumption that $f$
is $(\eps,\eps)$-quasirandom. \medskip

Since the $\mathsf{3XOR}_\delta$ test accepts $(\eps,\eps)$-quasirandom
functions with probability at most $\half + \half \sqrt{\eps}$ ({\it
  i.e.} $s = \lhalf$), we have shown that it is a
$(1-\frac{\delta}{2}, \lhalf)$ dictator versus quasirandom test.

\pagebreak
 
\section{CSPs and hardness of approximation}

\begin{center}
\large{Thursday, 1st March 2012} 
\end{center}

\subsection{Constraint satisfaction problems} 

We begin by noting that function testers can be viewed more generally
as string testers: the tester is given blackbox access to a string $w
\in \bits^N$ ({\it i.e.} the truth-table of $f$, so $N = 2^n$), and if
$w$ satisfies some property $P_1\sse \bits^N$ ({\it e.g.}
dictatorship) the tester accepts with probability say at least $\frac{2}{3}$,
and if $w$ satisfies some other property $P_2\sse \bits^N$ ({\it e.g.}
quasirandomness, far from dictatorship, {\it etc.}) it rejects with
probability at least $\frac{2}{3}$.\medskip

We may view (non-adaptive) string testers simply as a list of
instructions. For example, 
\bit
\itemsep -.5pt
\item[] with probability $p_1$ query $w_1,w_5,w_{10}$ and accept iff
  $\phi_2(w_1,w_5,w_{10})$ 
\item[] with probability $p_2$ query $w_{17},w_4,w_3$ and accept iff
  $\phi_8(w_{17},w_4,w_3)$
\item[] with probability $p_3$ query $w_2,w_{12},w_7$ and accept iff
  $\phi_4(w_2,w_{12},w_7)$ 
\item[] \quad \quad \quad $\vdots $ \eit Here $\phi_1,\phi_2,\ldots$
  are predicates $\bits^k\to\set{{\sf T},{\sf F}}$. From this
  point-of-view, we see that a string tester naturally defines a
  weighted constraint satisfaction problem (CSP) over a domain of $N$
  boolean variables, with the predicates $\phi_i$'s as constraints and
  the associated $p_i$'s as weights.  The question of determining
  which string $w\in\bits^N$ passes the test with highest probability
  is then equivalent to the question of finding an optimal assignment
  that satisfies the largest weighted fraction of predicates. \medskip

  Under this correspondence, an explicit $(c,s)$ dictator versus
  quasirandom test for functions $f:\bits^\ell\to\bits$ defines an
  explicit instance of a weighted CSP over $L = 2^\ell$ boolean
  variables. Since all $\ell$ dictators pass with probability at least
  $c$, there are $\ell$ special assignments each of which satisfy at
  least a $c$ weighted fraction of constraints. Furthermore, since all
  quasirandom functions pass the test with probability at most $s +
  o(1)$, any CSP assignment which is, roughly speaking, ``very
  unlike'' the $\ell$ special assignments will satisfy at most a
  $s+o(1)$ fraction of constraints. In other words, any CSP assignment
  that satisfies at least an $s + \Omega(1)$ fraction of constraints
  must be at least ``slightly suggestive'' of at least one of the
  $\ell$ special assignments (we say that $f$ is suggestive of the
  $i$-th coordinate if $\Inf_i^{(1-\eps)}(f) >  \eps$, and in
  particular, an $(\eps,\eps)$-quasirandom function suggests none of
  its coordinates).\medskip

  On Saturday Per will prove the following theorem establishing a
  formal connection between dictator versus quasirandom tests, the
  {\sc Unique-Label-Cover} problem, and the hardness of approximating
  certain CSPs \cite{Kho02, KR03,
  KKMO07, Aus08}:

  \bthm Suppose there is an explicit $(c,s)$ dictator versus
  quasirandom test that uses predicates $\phi_1,\ldots,\phi_r$. For
  every $\eps > 0$ there exists a polynomial-time reduction where:

\begin{center} 
\label{thm:bulccsp}
\begin{tabular}{ccc} 
{\sc Unique-Label-Cover}  & $\longrightarrow$ & CSP with constraints
$\phi_1,\ldots,\phi_r$ \vspace{5pt}  \\
{\sf YES} instance & $\longrightarrow$ & there exists an assignment
that satisfies \\
& & a $(c-\eps)$ fraction of constraints  \vspace{5pt} \\
{\sf NO} instance  & $\longrightarrow$ & every assignment satisfies at
most \\ & & an $(s+\eps)$ fraction of constraints\\ 
\end{tabular} 
\end{center} 
\ethm

The Unique Games Conjecture \cite{Kho02} asserts that approximating
the {\sc Unique-Label-Cover} problem is {\sf NP}-hard. Theorem
\ref{thm:bulccsp} therefore says that assuming the UGC, for any
constant $\eps > 0$ an explicit $(c,s)$ dictator versus quasirandom
test implies the {\sf NP}-hardness of $((s/c)+\eps)$-factor
approximating CSPs with constraints corresponding to the predicates
used by the test.\medskip

Recall that our analyses of all three dictator versus quasirandom
tests we have seen so far depend on the promise that $f$ is odd (in
particular, note that the $(0,0)$-quasirandom constant function $f
\equiv 1$ pass both the KKMO $\rho$-noise test and H\r{a}stad's
$\mathsf{3XOR}_\delta$ test with probability 1). One way to elide this
assumption is to view them as testers for general functions
$g:\bits^{\ell-1}\to\bits$ (corresponding to half of the truth table
of an odd function $f:\bits^{\ell}\to\bits$, say for the inputs with
$x_i = 1$) where the respective predicates allow literals instead of
just variables.  Although the string $w = 1^N$ ({\it i.e.} $f \equiv
1$) trivially satisfies any CSP with constraints of the form $w_i =
w_j$ ({\it i.e.} the KKMO $\rho$-noise test) or $w_iw_jw_k = 1$ ({\it
  i.e.} the H\r{a}stad $\mathsf{3XOR}_\delta$ test), the same is no
longer true if literals are allowed in the constraints.  Given this,
we may apply Theorem \ref{thm:bulccsp} to the $\NAE$ test, KKMO
$\rho$-noise test, and H\r{a}stad's $\mathsf{3XOR}_\delta$ test to
conclude that under the UGC, \bit \itemsep -.5pt
\item Approximating {\sf MAX-3NAE-SAT} to a factor of $0.91226\ldots +
  \eps$ is {\sf NP}-hard. 
\item Approximating {\sf MAX-2LIN} to a factor of
  $(1-\inv{\pi}\arccos(\rho))/(\half + \half \rho) + \eps$ is {\sf NP}-hard.
\item Approximating {\sf MAX-3LIN} to a factor of $(\half + \eps)$ is
  {\sf NP}-hard. 
\eit

\subsection{Berry-Ess\'een} 
\label{sec:bbe}  
In this section we prove the Berry-Ess\'een theorem
\cite{Ber41,Ess42}, a finitary version of the central limit theorem
with explicit error bounds.  Actually we will give a proof that only
yields a polynomially weaker error bound, the upshot being that the
proof is relatively simple and can be easily generalized to other
settings (as we will see tomorrow, the Mossel-O'Donnell-Olezkiewicz
proof of the invariance principle, an extension of the Berry-Ess\'een
theorem to low-degree polynomials, is very similar in spirit). We will
need Taylor's theorem:

\blem[Taylor]
Let $\psi$ be a smooth function and $r\in \N$. For all $x\in \R$ and
$\eps > 0$ there exists a $y\in [x,x+\eps]$ such that 
\[ \psi(x+\eps) = \psi(x) + \eps  \psi^{(1)}(x) + {\lfrac 1 {2!}}\cdot
\eps^2\psi^{(2)}(x) + \ldots + {\lfrac 1 {(r-1)!}}\cdot {\eps^{r-1}} \psi^{(r-1)}(x) +
{\lfrac 1 {r!}}\cdot{\eps^r} \psi^{(r)}(y).\] 
In particular, $\psi(x+\eps) = \psi(x) + \eps\cdot \psi^{(1)}(x) +
\lhalf \cdot \eps^2 \psi^{(2)}(x) + {\lfrac 1 6}\cdot
\eps^3 \psi^{(3)}(x) + \error$, 
where the error term has magnitude at most
$\norm{\psi^{(4)}}_\infty\cdot\eps^4/24$.  
\elem 

\bprop[hybrid argument]
\label{prop:bbehybrid} Let $X_1,\ldots,X_n$ be independent random
variables satisfying $\E[X_i] = 0$. Let $\sigma_i^2 := \E[X_i^2]$ and
suppose $\sumi\sigma_i^2 = 1$.  Let $\bfX = \sumi X_i$, $\calg\sim
N(0,1)$ and $\psi:\R\to\R$. Then \[ |\E[\psi(\bfX)]- \E[\psi(\calg)]|
= O\Big(\norm{\psi^{(4)}}_\infty \lsum_{i=1}^n \E[X_i^4]\Big).\] Note that if
each $X_i$ is $B$-reasonable then $\sumi\E[X_i^4] \leq B\cdot
\sumi\sigma_i^4 \leq B\cdot \max\set{\sigma_i^2}$.  \eprop

\bpf We will view $\calg$ as the sum of independent Gaussians $\calg_1
+ \ldots + \calg_n$, where each $\calg_i\sim N(0,\sigma_i^2)$. The
proof proceeds by a hybrid argument, showing that only a small
error is introduced whenever each $X_i$ in $\bfX$ is replaced by the
corresponding Gaussian.  More precisely, for each $i = 0,\ldots,n$, we
define the random variable $\bfZ_i := \calg_1 + \ldots + \calg_i +
 X_{i+1} +\ldots +  X_n$; these $n+1$ random variables
interpolate between $\bfZ_0 = \bfX$ and $\bfZ_n = \calg$.  We will
prove the inequality
\[  |\E[\psi(\bfZ_{i-1})] - \E[\psi(\bfZ_i)]| = O(\norm{\psi^{(4)}}_\infty\cdot
  \E[X_i^4]).\] 
for all $i\in [n]$, noting that this implies the theorem by the
triangle inequality. Fix $i\in [n]$ and define the random variable
${\bf R} := \calg_1 + \ldots + \calg_{i-1} + X_{i+1}
+\ldots + X_n$, so $\bfZ_{i-1} = {\bf R} +  X_i$ and
$\bfZ_i = {\bf R} + \calg_i$.  Our goal is therefore to
bound $|\E[\psi({\bf R} + \sigma_i\cdot X_i)] - \E[\psi({\bf R} +
\sigma_i\cdot \calg_i)]|$.  Applying Taylor's theorem twice, we get 
\begin{eqnarray*}
|\E[\psi(\bfZ_{i-1})] - \E[\psi(\bfZ_i)]| &=& \Big|\E\big[\psi({\bf
  R}) + X_i\cdot \psi^{(1)}({\bf R}) +
\lhalf X_i^2\cdot  \psi^{(2)}({\bf R}) +
\lfrac{1}{6} X_i^3\cdot \psi^{(3)}({\bf R}) +
  \error_1\big] \\
&-&  \E\big[\psi({\bf
  R}) + \calg_i\cdot \psi^{(1)}({\bf R}) +
\lhalf \calg_i^2\cdot \psi^{(2)}({\bf R}) +
{\lfrac 1 6}\calg_i^3\cdot  \psi^{(3)}({\bf R}) +
  \error_2\big]\Big| \\
&=& |\E[\error_1 - \error_2]|.
\end{eqnarray*} 
Here we have used that fact that ${\bf R}$ is independent of $X_i$ and
$\calg_i$, along with the assumption that $X_i$ and $\calg_i$ have
matching first and second moments.  Substituting bounds on the error
terms $\error_1$ and $\error_2$ given by Taylor's theorem, we complete
the proof:
\[ 
|\E[\error_1 - \error_2]| \leq \E\bigg[\frac{\norm{\psi^{(4)}}_\infty\cdot X_i^4}{24} +
\frac{\norm{\psi^{(4)}}_\infty\cdot \calg_i^4}{24}\bigg] 
  =    O(\norm{\psi^{(4)}}_\infty\cdot \E[X_i^4]). 
\] 
\epf 

The same proof can be rewritten to show that if $\bfY = Y_1 + \ldots +
Y_n$ is the sum of independent random variables satisfying $\E[X_i] =
\E[Y_i]$, $\E[X_i^2] = \E[X_i^2]$, and $\E[X_i^3] = \E[X_i^3]$ (the
matching moments property), then $|\E[\psi(\bfX)-\psi(\bfY)]| \leq
\norm{\psi^{(4)}}_\infty\cdot \sumi \E[X_i^4]+\E[Y_i^4]$. \medskip

Let $\psi_t:\R\to\R$ be the threshold function that takes value 1 if
$x < t$, and 0 otherwise.  Of course the $4^{th}$ derivative of this
function is not uniformly bounded, but note that if it were the
Berry-Ess\'een theorem would follow as an immediate corollary of
Proposition \ref{prop:bbehybrid}.  Instead, we will use the fact that
$\psi_t$ is well-approximated by a function which does have a
uniformly bounded $4^{th}$ derivative, which then implies a slightly
weaker version of the Berry-Ess\'een theorem.

\blem[smooth approximators of thresholds]
\label{lem:bapproxthreshold} 
Let $t\in\R$ and $0< \lambda < 1$. There exists a function
$\psi_{t,\lambda}:\R\to\R$ with $\norm{\psi^{(4)}_{t,\lambda}}_\infty = O(1/\lambda^4)$ that
approximates $\psi_t$ in the following sense: \ben \itemsep -.5pt
\item $\psi_{t,\lambda}(x) = \psi_{t}(x) = 1$ if $x < t - \lambda$. 
\item $\psi_{t,\lambda}(x) \in [0,1]$ if $x\in [t-\lambda,t+\lambda]$. 
\item $\psi_{t,\lambda}(x) = \psi_t(x) = 0$ if $x > t + \lambda$. 
\een 
\elem

We are now ready to prove a weak version of the Berry-Ess\'een
theorem. 

\bprop[weak Berry-Ess\'een] Let $X_1,\ldots,X_n$ be independent,
$B$-reasonable random variables satisfying $\E[X_i] =0$. Let
$\sigma_i^2 := \E[X_i^2]$, $\tau := \max\set{\sigma_i^2}$, and suppose
$\sumi\sigma_i^2 =1$. Let $S = X_1 + \ldots + X_n$ and $\calg\sim
N(0,1)$. For all $t\in \R$,
\[  |\Pr[S\leq t] - \Pr[\calg \leq t]| \leq
  O((B\tau)^{1/5}).\] 
\eprop 
 
\bpf Since $\psi_{t+\lambda,\lambda}(x) = 1$ for all $x < t$ we have
$\Pr[S \leq t] \leq
\E[\psi_{t+\lambda,\lambda}(S)].$ Now using the
fact that $\norm{\psi^{(4)}_{t+\lambda,\lambda}}_\infty= O(1/\lambda^4)$, we
apply Proposition \ref{prop:bbehybrid} to get 
\[ 
\E[\psi_{t+\lambda,\lambda}(S)] =
\E[\psi_{t+\lambda,\lambda}(\calg)] \pm O(B\tau/\lambda^4)
\] 
Since $\psi_{t + \lambda,\lambda}(x)$ is at most 1
for all $x \leq t+2\lambda$ and 0 otherwise, we have 
\[ \E[\psi_{t+\lambda,\lambda}(\calg)] \leq \Pr[\calg < t + 2\lambda]
= \Pr[\calg < t ] + O(\lambda),
\]  
and so combining both error bounds gives us
$\E[\psi_{t+\lambda,\lambda}(S)] \leq \Pr[\calg < t
] +  O(B\tau/\lambda^4) + O(\lambda)$.  Arguing symmetrically for
$\psi_{t-\lambda,\lambda}(x)$ gives us $|\Pr[S < t]
- \Pr[\calg < t]| = O(B\tau/\lambda^4) + O(\lambda)$, and taking
$\lambda = (B\tau)^{1/5}$ yields the claim.  \epf
  
\pagebreak 

\section{Majority Is Stablest}
 
\begin{center}
\large{Friday, 2nd March 2012} 
\end{center}

Our definition of $\rho$-correlated Gaussians (Definition
\ref{def:bcorrgauss}) extend naturally to higher dimensions: let
$\vec{\calg}$ and $\vec{\calg'}$ be independent standard
$n$-dimensional Gaussians ({\it i.e.} $\vec{\calg} =
(\calg_1,\ldots,\calg_n)$ where each $\calg_i\sim N(0,1)$ is an
independent standard Gaussian, and similarly for $\vec{\calg'}$). Then
$\vec{\calg}$ and $\vec{\calh} := \rho\cdot \vec{\calg} +
\sqrt{1-\rho^2}\cdot \vec{\calg'}$ are $\rho$-correlated Gaussians.
Just like in the one-dimension case, we have $\E[\calg_i\calh_i] =
\rho$ for all $i\in [n]$.

\bdefn[Gaussian noise stability]
Let $f:\R^n\to\R$ and $\rho\in [-1,1]$. The Gaussian noise stability
of $f$ at noise rate $\rho$ is 
\[ \GStab_\rho(f) :=
\E[f(\vec{\calg})f(\vec{\calh})], \quad \text{where
  $\vec{\calg},\vec{\calh}$ are $\rho$-correlated Gaussians}.  \] 
\edefn

We begin by showing that $\GStab_\rho(f) = \Stab_\rho(f)$ for
multilinear polynomials $f:\R^n\to\R$.  To see this, let $f(x) =
\sumalls c_S\prod_{i\in S}x_i$ and note that
\[ 
 \E\Big[\Big(\lsum_{S\sse [n]}c_S\lprod_{i\in
  S}\vec{\calg_i}\Big)\Big(\lsum_{T\sse [n]}c_T\lprod_{i\in
  T}\vec{\calh_i}\Big)\Big] = \sum_{S,T\sse [n]}c_Sc_T \E\Big[
\lprod_{i\in S}\calg_i\lprod_{i\in T}\calh_i\Big] 
=  \sumalls c_S^2 \E\Big[ \lprod_{i\in S}\calg_i\calh_i \Big],
\] 
where we have used the independence of $\calg_i$ from
$\calg_j,\calh_j$ for $j\neq i$, along with the fact that $\E[\calh_i]
= \E[\calg_i] = 0$. Now again by independence and the fact that
$\E[\calg_i\calh_i] = \rho$ we conclude that $\GStab_\rho(f) =
\sumalls \rho^{|S|}c_S^2$, which agrees with the formula for $\Stab_\rho(f)$
we derived in Proposition \ref{prop:bnoise}.

\subsection{Borell's isoperimetric inequality}

\bthm[\cite{Bor85}] Let $f:\R^n\to\bits$ with $\E[f(\vec{G})] = 0$,
where $\vec{\calg}$ is a standard $n$-dimensional Gaussian. Let $0\leq
\rho\leq 1$.  Then $\GStab_\rho(f) \leq 1
-\frac{2}{\pi}\arccos(\rho)$.  \ethm

In this section we present Kindler and O'Donnell's recent simple proof
of (a special case of) Borell's theorem \cite{KO12}. We first
introduce a few definitions and give a geometric interpretation of the
theorem as an isoperimetric inequality in multidimensional Gaussian
space. 

\bdefn[rotation sensitivity]
Let $f:\R^n\to\bits$ and $\delta \in [0,\pi]$. The rotation
sensitivity of $f$ at $\delta$ is defined to be $\RS_f(\delta) :=
\Pr[f(\vec{\calg})\neq f(\vec{\calh})]$, where $\vec{\calg}$ and
$\vec{\calh}$ are $\cos(\delta)$-correlated Gaussians.  
\edefn

Recall that $\vec{\calg}$ and $\vec{\calh}$ are
$\cos(\delta)$-correlated if $\calh = \cos(\delta)\cdot \vec{\calg} +
\sin(\delta)\cdot \vec{\calg'}$ where $\vec{\calg'}$ is a standard
$n$-dimensional Gaussian independent of $\vec{\calg}$; we typically we
think of $\delta$ as small, so $\cos(\delta)\approx 1-\half\delta^2$
is close to 1 and $\sin(\delta)\approx \delta$ is a small quantity.
If we view $f:\R^n\to\bits$ as the indicator of a subset $\indi_f$ of
$\R^n$, the quantity $\RS_f(\delta)$ measures the probability that the
set $\indi_f$ separates a random Gaussian vector $\calg$ from a noisy
copy of it (roughly speaking, $\calg$ with $\sin(\delta)\cdot
\vec{\calg'}$ of noise added).  Therefore we may think of the rotation
sensitivity of $f$ as a measure of the boundary size of $\indi_f$.
Indeed, Kindler and O'Donnell show that $\lim\sup_{\delta\to 0^+}
\RS_f(\delta)/\delta$ is within a constant factor of the traditional
definition of the Guassian surface area of $\indi_f$ for
``sufficiently nice'' sets $\indi_f$.\medskip

Since $\RS_f(\delta) = \Pr[f(\vec{\calg})\neq f(\vec{\calh})] = \half
- \half\cdot \GStab_{\cos(\delta)}(f)$, Borell's theorem can be
equivalently stated as $\RS_f(\delta)\geq \frac{\delta}{\pi}$ for
functions $f$ with $\E[f(\vec{\calg})] = 0$; it gives a lower bound on
the boundary size of sets $\indi_f$ with Gaussian volume $\half$.
Sheppard's formula tells us that if $f = \sgn(a_1x_1 + \ldots +
a_nx_n)$ then $\Pr[f(\vec{\calg}) \neq f(\vec{\calh})] =
\inv{\pi} \arccos(\cos(\delta)) = \frac{\delta}{\pi}$, and so Borell's
inequality is tight when $\indi_f$ is any halfspace defined by a
hyperplane passing through the origin.\medskip

Kindler and O'Donnell prove Borell's theorem for $\delta =
\frac{\pi}{2\ell}$ where $\ell\in\N$ (we may assume $\ell \geq 2$
since the inequality is trivially true for uncorrelated Gaussians).
As a warm-up, we consider the case of $\delta = \frac{\pi}{4}$ ({\it
  i.e.} $\ell = 2$). Let $f:\R^n\to\bits$ with $\E[f(\vec{\calg})] =
0$ and note that our goal is to prove $\RS_f(\frac{\pi}{4}) \geq
{\lfrac 1 4}$. Let $\vec{\calg}$ and $\vec{\calg'}$ be independent
standard $n$-dimensional Gaussians and set $\vec{\calh} =
\inv{\sqrt{2}}\cdot\vec{\calg} + \inv{\sqrt{2}}\cdot
\vec{\calg'}$. Note that $(\vec{\calg},\vec{\calh})$ and
$(\vec{\calg'},\vec{\calh})$ are both $(\inv{\sqrt{2}})$-correlated
Gaussians, and so we have
\[ \Pr[f(\vec{\calg})\neq f(\vec{\calh})] + \Pr[f(\vec{\calh}) +
f(\vec{\calg'})] = 2\cdot \RS_f({\lfrac {\pi}{4}}). \] By a union
bound, the quantity on the left hand side is at most
$\Pr[f(\vec{\calg}) \neq f(\vec{\calg'})]$, and since $f$ is balanced
this probability is exactly $\inv{2}$ and the proof is complete.  We
remark that an identical proof can be carried out for bounded
functions $f:\R^n\to [-1,1]$, with $\RS_f(\delta) := \E[\half -\half
f(\vec{\calg})f(\vec{\calh})]$ as the generalized definition of
rotation sensitivity; all the steps remain the same, except that the
inequality $(\half-\half ac) \leq (\half-\half ab) + (\half -\half
bc)$ for all $a,b,c\in [-1,1]$ will be used in place of the union
bound.

\bthm[\cite{KO12}]
Let $f:\R^n\to\bits$ with $\E[f(\vec{\calg})] = 0$, where
$\vec{\calg}$ is a standard $n$-dimensional Gaussian. Let $\delta =
\frac{\pi}{2\ell}$ for some $\ell\in\N$, $\ell\geq 2$. Then
$\RS_f(\delta)\geq \inv{2\ell}$. 
\ethm

\bpf Let $\vec\calg$ and $\vec{\calg'}$ be independent standard
$n$-dimensional Gaussians. For $j = 0,\ldots,\ell$, we define the
hybrid random variable $\vec\calg^{(j)} := \cos(j\delta)\cdot\vec\calg
+ \sin(j\delta)\cdot \vec\calg'$, and note that they interpolate
between $\vec\calg^{(0)} = \vec\calg$ and $\vec\calg^{(\ell)} =
\vec\calg'$. Next, we claim that $\vec\calg^{(j-1)}$ and
$\vec\calg^{(j)}$ are $\cos(\delta)$-correlated for every $j\in
[\ell]$. Since the $n$ coordinates are independent, it suffices to
consider the first coordinates of $\vec\calg^{(j-1)}_1$ and
$\vec\calg^{(j)}_1$ and show that they are $\cos(\delta)$-correlated. We
write $\calg^{(j-1)}$ for ${\vec\calg^{(j-1)}}_1 =
\cos((j-1)\delta)\cdot \calg + \sin((j-1)\delta)\cdot\calg'$ and
$\calg^{(j)}$ for $\vec\calg_1^{(j)} = \cos(j\delta)\cdot\calg +
\sin(j\delta)\cdot\calg'$, and expand the expectation of their product
to check that
\begin{eqnarray*}
\E\big[\calg^{(j-1)}\calg^{(j)}\big] &=&
\cos((j-1)\delta)\cos(j\delta)\E[\calg\calg] \\ 
& & + \cos((j-1)\delta)\sin(j\delta)\E[\calg\calg'] \\
& & + \sin((j-1)\delta)\cos(j\delta)\E[\calg'\calg] \\
& & + \sin((j-1)\delta)\sin(j\delta)\E[\calg'\calg']\\  &=& 
\cos((j-1)\delta)\cos(j\delta) + \sin((j-1)\delta)\sin(j\delta) \\
&=& \cos((j-1)\delta - j\delta)\ =\ \cos(\delta). 
\end{eqnarray*} 
Here we have used the trigonometric identity $\cos(\theta)cos(\theta')
+ \sin(\theta)\sin(\theta') = \cos(\theta-\theta')$. Since
$\vec\calg^{(j-1)}$ and $\vec\calg^{(j)}$ are
$\cos(\delta)$-correlated, we apply the union bound to get 
$\lhalf = \Pr[f(\vec\calg)\neq f(\vec\calg')] \leq \sum_{i=1}^\ell
\Pr[f(\vec\calg^{(\ell-1)})\neq f(\vec\calg^{(\ell)})] = \ell\cdot
\RS_f(\delta)$, and the proof is complete.  \epf

\subsection{Proof outline of MIST} 
  
In this section we sketch the proof of the Majority Is Stablest
theorem (MIST):

\bthm[\cite{KKMO07,MOO10}] Let $f\isafunc$ and $\eps,\rho >
0$. Suppose $\E[f] = 0$ and $\Inf_i(f) \leq \eps$ for all $i\in
[n]$. Then $\Stab_\rho(f) \leq 1-\frac{2}{\pi}\arccos(\rho) +
o_\eps(1)$.  \ethm

{\bf Step 1.}  First consider $T_{1-\gamma}f$ for some small $\gamma > 0$.
  Note that \bit
\itemsep -.5pt
\item[(a)] $\Inf_i(T_{1-\gamma}f) \leq \Inf_i(f) \leq  \eps$ for all $i\in [n]$. 
\item[(b)]  $T_{1-\gamma}f$ is bounded
  since $T_{1-\gamma}$ is an averaging
  operator.
\item[(c)] $\Stab_\rho(f) =
  \Stab_{\rho'}(T_{1-\gamma}f)$ where $\rho' :=
  \frac{\rho}{(1-\gamma)^2}$. 
\eit

{\bf Step 2.} Next we truncate $T_{1-\gamma}f$ to get $g :=
(T_{1-\gamma}f)^{\leq k} = \sum_{|S|\leq
  k}(1-\gamma)^{|S|}\hatfs\chis$, where $k := \poly(\inv{\gamma})$.
Note that $g$ has low-degree but it may no longer
bounded. Nevertheless we can say that
  \[ \norm{(T_{1-\gamma}f)- g}_2^2 = \sum_{|S| > k} (1-\gamma)^{2\cdot
    |S|}\hatfs^2 \leq (1-\gamma)^{2k} \sum_{|S| > k}\hatfs^2 \leq
  \gamma, \] and since $g$ is bounded this implies
  $\E[\sqdist_{[-1,1]}(g(X_1,\ldots,X_n))] \leq \gamma$. For the same
  reason, $\Stab_{\rho'}(T_{1-\gamma}f) \leq \Stab_{\rho'}(g) +
  \gamma$.  Here $\sqdist_{[-1,1]}:\R\to\R^{\geq 0}$ is the function
  that gives the squared distance to the interval $[-1,1]$; {\it i.e.}
  $\sqdist_{[-1,1]}(t) = 0$ if $t\in [-1,1]$, and $(|t|-1)^2$
  otherwise. \bigskip

{\bf Step 3.}  We apply the invariance principle (an extension of
the Berry-Ess\'een theorem to low-degree multilinear polynomials,
proved in the next section) to $g$ and the test function
$\sqdist_{[-1,1]}$ to bound
\[ \E[\sqdist_{[-1,1]}(g(\vec\calg))] \leq \E[\sqdist_{[-1,1]}(g(\vec
X))] + \poly(2^k,\eps) = \gamma + \poly(2^k,\eps). \] We are omitting
a few details here since the invariance principle requires test
functions to have uniformly bounded $4^{th}$ derivatives, just like in
Berry-Ess\'een, so we actually need a smooth approximation of
$\sqdist_{[-1,1]}$.\bigskip 

{\bf Step 4.} Finally we consider $g':\R^n\to [-1,1]$, the truncation
of $g$ to the interval $[-1,1]$; {\it i.e.} $g'(u) = g(u)$ if $g(u)\in
[-1,1]$, and $\sgn(g(u))$ otherwise.  We have $\Stab_{\rho'}(g) =
\GStab_{\rho'}(g)$ since $g$ is multilinear, and $\GStab_{\rho'}(g')
\leq 1-\frac{2}{\pi}\arccos(\rho')$ by Borell's theorem (again we are
eliding some details here since $g'$ may not satisfy $\E[g'(\calg)] =
0$). It remains to argue that $\GStab_{\rho'}(g) \approx
\GStab_{\rho'}(g')$, and for this we need to define the
Ornstein-Uhlenbeck operator $U_\rho$, the Gaussian analogue of the
noise operator $T_\rho$: \bdefn[Ornstein-Uhlenbeck] Let $f:\R^n\to\R$
and $\rho\in [-1,1]$. The Ornstein-Uhlenbeck operator $U_\rho$ acts of
$f$ as follows: $(U_\rho f)(x) := \E[f(\rho\cdot x +
\sqrt{1-\rho^2}\cdot \vec\calg)]$, where $\vec\calg$ is a standard
$n$-dimensional Gaussian.  \edefn With this definition, we may express
Gaussian noise stability as $\GStab_\rho(f) =
\E[f(\vec\calg)U_\rho(\vec\calg)]$ and compute
\begin{eqnarray}
  |\GStab_{\rho'}(g)-\GStab_{\rho'}(g')| &=& |\E[g\cdot U_\rho
  g  - g'\cdot U_\rho g' ]| \nonumber \\
  &\leq& |\E[g\cdot U_\rho
  g  - g' \cdot U_\rho g ]| +
  |\E[g'\cdot U_\rho g -g'\cdot U_\rho
  g' ]| \nonumber \\
  &=& |\E[(g-g') \cdot U_\rho g ]| + 
 |\E[(g-g') \cdot U_\rho g' ]|\label{eq:bmist1} \\
&=& \big(\E[(g-g') ^2]\big)^{1/2}\big(\E[(U_\rho
g)^2]\big)^{1/2} \nonumber \\
& & + \ \big(\E[(g-g') ^2]\big)^{1/2}\big(\E[(U_\rho
g')^2]\big)^{1/2} \label{eq:bmist2} \\
&\leq& 2\,\big(\E[(g-g')^2]\big)^{1/2}.\label{eq:bmist3}
\end{eqnarray}
Here (\ref{eq:bmist1}) holds since $\E[g'\cdot U_\rho g] =\E[U_\rho g'
\cdot g]$, (\ref{eq:bmist2}) is an application of Cauchy-Schwarz, and
(\ref{eq:bmist3}) uses the fact that $U_\rho$ is a contraction on
$L^2$. Finally we note that $\E[(g-g')^2]$ is simply
$\E[\sqdist_{[-1,1]}(g)]$ and the proof is complete. 

\subsection{The invariance principle} 
 
In this section we prove (a special case of) the
Mossel-O'Donnell-Olezkiewicz invariance principle \cite{MOO10} for
multilinear polynomials with low influences and bounded degree; in
full generality the principle states that the distribution of such
polynomials is essentially invariant for all product spaces.  The crux
of the proof is a low-degree analogue of Proposition
\ref{prop:bbehybrid}; once again we proceed by a hybrid argument,
showing that a small error is introduced whenever we replace a
Rademacher random variable with a standard Gaussian. This is sometimes
known as the Lindeberg replacement trick, first appearing in
Lindeberg's proof of the central limit theorem \cite{Lin22}.  There
has been other work generalizing Lindeberg's argument to the
non-linear case \cite{Rot75, Rot79, Cha06}, but these results either
yield weaker error bounds or require stronger conditions ({\it e.g.}
worst-case influences rather than average-case).
 
\bprop[hybrid argument]
\label{prop:bmoothm318}
Let $Q$ be a degree-$d$ multilinear polynomial $Q(u) =
\sum_{|S|\leq d} c_S\prod_{i\in S}u_i$ and assume: \ben \itemsep -.5pt
\item The coefficients $c_S\in\R$ are normalized to satisfy
$\sum_{S\neq\emptyset}c_S^2 = 1$. 
\item We write $\tau_i$ to denote $\Inf_i(Q) = \sum_{S\ni i}c_S^2$,
  and let $\tau = \max_{i\in [n]}\Inf_i(Q)$. 
\item $\psi:\R\to\R$ is a function satisfying $|\psi^{(4)}(x)| \leq C$
  for all $x\in\R$.
\item $\bfX = Q(X_1,\ldots,X_n)$ and $\bfY =
  Q(\calg_1,\ldots,\calg_n)$ where $X_1,\ldots,X_n$ are independent
  Rademachers and $\calg_1,\ldots,\calg_n$ are independent standard
  Gaussians.  \een Then $|\E[\psi(\bfX)] - \E[\psi(\bfY)]| \leq d\cdot
  9^d\cdot C\cdot \tau$.  \eprop

  \bpf We first define a sequence of hybrid random variables that
  interpolate between $\bfX$ and $\bfY$. For each $i = 0,\ldots,n$ we
  define the random variable $\bfZ_i =
  Q(\calg_1,\ldots,\calg_i,X_{i+1},\ldots,X_n)$, and note that $\bfZ_0
  = \bfX$ and $\bfZ_n = \bfY$.  As before, it suffices to prove
\begin{equation} |\E[\psi(\bfZ_{i-1})]-\E[\psi(\bfZ_i)]| \leq C\cdot 9^d\cdot
\tau_i^2\label{eq:bmoo1}
\end{equation}
for all $i\in [n]$. Note that the overall
claim follows from the above by telescoping, the triangle
inequality, and the fact that
\[ \sumi\tau_i^2 \leq \tau \cdot \sumi\tau_i = \tau\cdot
\sumi\sum_{S\ni i} c_S^2 = \tau\cdot \sum_{|S|\leq d}|S|\cdot c_S^2
\leq \tau\cdot d. \] Here in the final inequality we have used our
assumption that the coefficients are normalized to satisfy
$\sum_{S\neq\emptyset}c_S^2 = \Var(\bfX) = \Var(\bfY) = 1$.  It
remains to prove (\ref{eq:bmoo1}).  Fix $i\in [n]$ and first
express $Q(u_1,\ldots,u_n)$ as the sum of two polynomials $R$ and $S$,
the former comprising all terms not containing $u_i$, and the latter
the rest with $u_i$ factored out.  That is,
\[ Q(u_1,\ldots,u_n) = R(u_1,\ldots,u_{i-1},u_{i+1},\ldots,u_n) + u_i
\cdot S(u_1,\ldots,u_{i-1},u_{i+1},\ldots,u_n),\] where $R$ has degree
at most $d$, and $S$ at most $d-1$ (note that if $d=1$ then $S$ is
simply the coefficient $\alpha_i$ of $u_i$ in the linear polynomial
$L$).  Next define the random variables
\begin{eqnarray*}
\bfR &=& R(\calg_1,\ldots,\calg_{i-1},X_{i+1},\ldots,X_n) \\
\bfS &=& S(\calg_1,\ldots,\calg_{i-1},X_{i+1},\ldots,X_n), 
\end{eqnarray*}  
and note that $\bfZ_{i-1} = \bfR + X_i\cdot \bfS$ and $\bfZ_i = \bfR +
\calg_i\cdot \bfS$.  We bound $|\E[\psi(\bfR + X_i\cdot\cals)] -
\E[\psi(\bfR + \calg_i\cdot\bfS)]|$ by considering their Taylor
expansions: \bigskip

\quad \quad $|\E[\psi(\bfZ_{i-1})] - \E[\psi(\bfZ_i)]|$
\[ =\  \Big|\E\big[\psi({\bf
  R}) + (X_i\cdot\bfS)\psi^{(1)}({\bf R}) +
{\lfrac 1 2} (X_i\cdot\bfS)^2 \psi^{(2)}({\bf R}) +
{\lfrac 1 6} (X_i\cdot\bfS)^3 \psi^{(3)}({\bf R}) +
  \error_1\big] \] \vspace{-18pt}
\[ \ \  - \  \E\big[\psi({\bf
  R}) + (\calg_i\cdot\bfS)\psi^{(1)}({\bf R}) +
{\lfrac 1 2} (\calg_i\cdot\bfS)^2 \psi^{(2)}({\bf R}) +
{\lfrac 1 6} (\calg_i\cdot\bfS)^3 \psi^{(3)}({\bf R}) +
  \error_2\big]\Big|. \]

  Note that the first four terms cancel out since $X_i$ and $\calg_i$
  are independent of $\bfS$ and $\bfR$, and the random variables $X_i$
  and $\calg_i$ have matching first, second and third
  moments. Applying the bounds on the error terms given by Taylor's
  theorem, we see that
\begin{eqnarray}
 |\E[\error_1-\error_2]| &\leq&
{\lfrac 1 {24}}\cdot C \E[X_i^4\cdot \bfS^4] +
{\lfrac 1 {24}}\cdot C \E[\calg_i^4\cdot\bfS^4] \label{eq:bmoo2} \\
&=& {\lfrac 1 {24}}\cdot {C} \E[\bfS^4] + {\lfrac 3 {24}}\cdot C \E[\bfS^4]
\label{eq:bmoo3} \\
&<& C\cdot \E[\bfS^4].  \nonumber 
\end{eqnarray} 
Here (\ref{eq:bmoo2}) uses our assumption that $\psi^{(4)}$ is
uniformly bounded by $C$, and (\ref{eq:bmoo3}) is by independence
along with the fact that $\E[\calg_i^4] = 3$.  Next, since $\bfS$ is a
degree-$d$ polynomial, we may apply Bonami's lemma (Theorem
\ref{thm:bbonami}) to get $C\cdot \E[\bfS^4] \leq C\cdot 9^d\cdot
\E[\bfS^2]^2$ and it remains to argue that $\E[\bfS^2]$ upper bounded
by $\tau_i = \Inf_i(Q) = \sum_{S\ni i}c_S^2$. Indeed, recall that $S$
is the polynomial comprising all terms of $Q$ containing $u_i$ with
$u_i$ factored out, and so
\[ \E[\bfS^2] = \E\Big[\big(\lsum_{S\ni i}c_S \lprod_{j\in S\backslash
  i} Y_j\big)^2\Big] = \sum_{S\ni i} c_S^2 = \tau_i, \] where each
$Y_j$ is either a Rademacher or standard Gaussian random variable,
depending on whether $j < i$.  We have shown that
$|\E[\psi(\bfZ_{i-1})] - \E[\psi(\bfZ_i)]| \leq C\cdot 9^d\cdot
\tau_i^2$, and the proof is complete.  \epf

In the proof of the Berry-Ess\'een theorem we needed the fact that the
anti-concentration of a standard Gaussian $\calg$ at radius $\eps$ is
$O(\eps)$.  That is, for all $t\in \R$ we have $\Pr[|\calg -t| < \eps]
= O(\eps)$.  The low-degree analogue of this small ball probability
is given by the following proposition due to Carbery and Wright
\cite{CW01, Kan11}.
 
\blem[Carbery-Wright]
There exists a universal constant $C$ such that the following
holds. Let $Q$ be a multilinear polynomial of degree $d$ over
$\calg_1,\ldots,\calg_n$, a sequence of independent standard
Gaussians, and $\eps > 0$. Then
\[ \Pr[|Q(\calg_1,\ldots,\calg_n)| \leq \eps] \leq C\cdot d\cdot
(\eps/\norm{Q(\calg)}_2)^{1/d}.\] In particular, if the coefficients
of $Q$ are normalized to satisfy $\Var(Q) = 1$ then for all $t\in\R$
and $\eps > 0$ we have $\Pr[|Q(\calg_1,\ldots,\calg_n)-t|\leq \eps] = O(d\cdot
\eps^{1/d})$. 
\elem

\blem[smooth test functions] Let $r\geq 2$ be an integer. There exists a
constant $B_r$ such that for all $0 < \lambda \leq 1/2$ and $t\in\R$
there exists a function $\Delta_{\lambda,t}:\R\to\R$ satisfying the
following:
\ben
\itemsep -.5pt
\item $\Delta_{\lambda,t}$ is smooth and $\norm{(\Delta_{\lambda,t})^{(r)}}_\infty \leq B_r\cdot
  \lambda^{-r}$. 
\item $\Delta_{\lambda,t}(x) = 1$ for all $x \leq t-2\lambda$. 
\item $\Delta_{\lambda,t}(x) \in [0,1]$ for all $x\in
  (t-2\lambda,t+2\lambda)$. 
\item $\Delta_{\lambda,t}(x) = 0$ for all $x\geq t+2\lambda$. 
\een 
 \elem

We are now ready to prove the invariance principle. 

\bthm[invariance] Let $Q(u_1,\ldots,u_n) = \sumalls c_S\prod_{i\in
  S}u_i$ be a degree-$d$ multilinear polynomial and assume
 \ben \itemsep -.5pt
\item The coefficients $c_S\in\R$ are normalized to satisfy
$\sum_{S\neq\emptyset}c_S^2 = 1$. 
\item We write $\tau_i$ to denote $\Inf_i(Q) = \sum_{S\ni i}c_S^2$,
  and let $\tau = \max_{i\in [n]}\Inf_i(Q)$. 
\item $\bfX = Q(X_1,\ldots,X_n)$ and $\bfY =
  Q(\calg_1,\ldots,\calg_n)$ where $X_1,\ldots,X_n$ are independent
  Rademachers and $\calg_1,\ldots,\calg_n$ are independent standard
  Gaussians. 
\een 
 Then for all $t\in \R$, $|\Pr[Q(X_1,\ldots,X_n) \leq t] -
\Pr[Q(\calg_1,\ldots,\calg_n)\leq t]| = O\big(d\cdot
(10^d\cdot\tau)^{1/(4d+1)}\big)$. 
\ethm

\bpf Let $t\in\R$.  We will write $Q(X)$ to denote
$Q(X_1,\ldots,X_n)$, $Q(\calg)$ for $Q(\calg_1,\ldots,\calg_n)$, and
$\psi$ for $\Delta_{\lambda,t+2\lambda}$ ($r = 4$), for some value of
$\lambda > 0$ to be determined later.  Recall that $\psi(x) = 1$ for
all $x \geq (t + 2\lambda)-2\lambda = t $ and so $\Pr[Q(X) \leq t]
\leq \E[\psi(Q(X))].$ Since $\psi^{(4)}$ is uniformly bounded by
$O(1/\lambda^4)$, we apply Proposition \ref{prop:bmoothm318} to get
that
\begin{eqnarray} \E[\psi(Q(X))] &\leq& \E[\psi(Q(\calg))] +
O(10^d\cdot\tau\cdot \lambda^{-4}).\nonumber \\
&\leq& \Pr[Q(\calg)\leq t + 4\lambda] + O(10^d\cdot \tau\cdot
\lambda^{-4}). \label{eq:moo4}\\
&=& \Pr[Q(\calg) \leq t] + \Pr[Q(\calg) \in (t,t+4\lambda)] +
O(10^d\cdot \tau\cdot \lambda^{-4}) \nonumber \\
&=& \Pr[Q(\calg) \leq t] + O(d\cdot (4\lambda)^{1/d}) +O(10^d\cdot
\tau\cdot \lambda^{-4}).\label{eq:moo5} 
\end{eqnarray} 
Here (\ref{eq:moo4}) is again by the properties of $\psi$, this time
using the fact that $\psi(x) = 0$ for all $x\geq (t+2\lambda) +
2\lambda$, and (\ref{eq:moo5}) is by Carbery-Wright. Choosing $\lambda
= (10^d\cdot \tau)^{d/(4d+1)}$, we have shown that 
\[ \E[\psi(Q(X))] \leq \Pr[Q(\calg)\leq t] + O\big(d\cdot
(10^d\cdot\tau)^{1/(4d+1)}\big).\] 
A symmetric argument establishes the analogous lower bound on
$\E[\psi(Q(X))]$, and this completes the proof. 
 \epf

\pagebreak 
 
\section{Testing dictators and {\sf UGC}-hardness}

\begin{center}
\large{Saturday, 3rd March 2012} \vspace{4pt} \\
\large{Guest lecture by Per Austrin}
\end{center}

\bdefn[unique label cover] Let $L$ be a positive integer. An instance
$\Psi$ of the $L$-{\sc Unique-Label-Cover} problem is a graph $G =
(V,E)$ where each edge $e\in E$ has an associated constraint that is a
permutation $\pi_e:[L]\to [L]$. A labelling of $\Psi$ is a assignment
to the vertices $\ell:V\to [L]$.  We say that an edge $(x,y)$ is
satisfied by $\ell$ if $\ell(x) = \pi_e(\ell(y))$, and the value of
$\ell$ is the fraction of edges satisfied by $\ell$.  The optimum of
$\Psi$, denoted $\opt(\Psi)$, is the maximum value of the optimal
assignment $\ell$.  \edefn

The {\sc Unique-Label-Cover} problem is a special case of the {\sc
  Label-Cover} problem where the constraints $\pi_e:[L]\to [L]$ are
not required to be permutations.  In particular, in the {\sc
  Unique-Label-Cover} problem assigning a label to a vertex
necessarily determines the labels of all its neighbors, whereas this
is not the case for the {\sc Label-Cover} problem.  Consequently, for
{\sc Unique-Label-Cover} the task of deciding whether there is an
assignment that satisfies all the edges ({\it i.e.} distinguishing
$\opt(\Psi) = 1$ versus $\opt(\Psi) < 1$) is easy: assume a label for
a vertex $v$ and deduce the labels for the remaining vertices in a
breadth-first fashion. If there is a conflict at some vertex we choose
another label for $v$ and repeat the same process. If no consistent
labeling can be found after iterating through all $L$ possible labels
for $v$ then $\opt(\Psi) < 1$; otherwise $\opt(\Psi) = 1$.  This is in
sharp contrast to the situation for {\sc Label-Cover}: it is known
that for every $\eps > 0$ there is an $L$ such that it is {\sf
  NP}-hard to distinguish between $\opt(\Psi) = 1$ versus $\opt(\Psi)
< \eps$ where $\Psi$ is an instance of $L$-{\sc Label-Cover}; we
sometimes refer to this as the $(1,\eps)$-hardness of {\sc
  Label-Cover}. \medskip

The Unique Games Conjecture of S. Khot \cite{Kho02} asserts that the
{\sc Unique-Label-Cover} problem is nevertheless very hard to
approximate as soon as we move to almost-satisfiable instances.

\begin{conjecture}[unique games]
For every $\eps > 0$ there exists an $L$ such that the it is {\sf
  NP}-hard to distinguish between $\opt(\Psi) \geq 1-\eps$ versus
$\opt(\Psi) < \eps$, where $\Psi$ is an instance of $L$-{\sc
  Unique-Label-Cover}.  Equivalently, for every $\eps > 0$ there exists an
$L$ such that the $L$-{\sc Unique-Label-Cover} problem is
$(1-\eps,\eps)$-hard. 
\end{conjecture}

Recent work of S. Arora, B. Barak and D. Steurer \cite{ABS10} gives an
algorithm for {\sc Unique-Label-Cover} running in time
$\exp(n^{\poly(\eps)})$. \medskip

Today we will prove the following theorem showing how explicit $(c,s)$
dictator versus quasirandom tests yield Unique Games-based hardness
results for certain constraint satisfaction problems \cite{Kho02, KR03,
  KKMO07, Aus08}:

\bthm Suppose we have a $(c,s)$ dictator versus quasirandom using
predicates from a set $T$. For every $L$ there exist a polynomial time
reduction $R$ from $L$-{\sc Unique-Label-Cover} to $\text{\sf
  MAX-CSP}(T)$ such that for every instance $\Psi$ of $L$-{\sc
  Unique-Label-Cover} and every $\eps > 0$, there exists a $\delta >
0$ satisfying \ben \itemsep -.5pt
\item (Completeness) If $\opt(\Psi) \geq 1-\delta$ then $\opt(R(\Psi)) \geq c-\eps$.
\item (Soundness) If $\opt(\Psi) < \delta$ then $\opt(R(\Psi)) < s + \eps$. 
\een 
\ethm 

First, a small catch: the dictator versus quasirandom test have to
work not only for boolean functions but also for bounded
functions $f:\bn\to [-1,1]$.  We may view any predicate
$\phi:\bits^k\to\zbits$ as $\phi^* : [-1,1]^k \to [0,1]$, where
$\phi^*(y_1,\ldots,y_k) :=
\E[\phi(\boldsymbol{x_1},\ldots,\boldsymbol{x_k})]$, the expectation
taken with respect to $\bits$-valued random variables
$\boldsymbol{x_i}$ satisfying $\E[\boldsymbol{x_i}] = y_i$.  It is
easy to check that $\phi^*(x) = \phi(x)$ for all $x\in\bits^k$, and in
fact we have $\phi^*(y_1,\ldots,y_k) = \sum_{S\sse
  [k]}\hat{\phi}(S)\prod_{i\in S}y_i$.  With this observation any
tester for boolean functions using predicate $\phi$ can be
extended to one for all bounded functions: instead of
accepting iff $\phi(f(x_1),\ldots,f(x_k)) = 1$, the tester accepts
with probability $\phi^*(f(x_1),\ldots,f(x_k))\in [0,1]$. 

\subsection*{The reduction from {\sc Unique-Label-Cover} to {\sf
    MAX-CSP}}

With this caveat out of the way, we are now ready to describe the
reduction $R$:  

\begin{center}
\begin{boxedminipage}{5.4in}
\begin{minipage}{5.4in}
\vspace{8pt}
\begin{quote}
  Let $\Psi$ be an instance of $L$-{\sc Unique-Label-Cover} defined
  over a graph $G = (V,E)$.  Suppose we have a $(c,s)$ dictator versus
  quasirandom test for functions $\bits^L\to [-1,1]$ using $k$-ary
  predicates from a set $T$.  Consider the following instance
  $R(\Psi)$ of $\textsf{MAX-CSP}(T)$:\medskip

{\bf Variables.} There will be $|V|\cdot 2^L$ variables: for each
$u\in V$ we define $2^L$ boolean variables $\{Z_{u,x} \stc
x\in\bits^L\}$.    \medskip

{\bf Constraints.} A random constraint will be sampled as follows: 
\vspace{-3pt}
\ben
\item Pick $u\in V$ uniformly. 
\item Pick $k$ neighbors $v_1,\ldots,v_k\in N(u)$ of
  $u$ uniformly independently. 
\item Define  $\widetilde{f_{u,v_i}} := f_{v_i}(x\circ
  \pi_{u,v_i})$.  
\item Pick $x_1,\ldots,x_k \in\bits^L$ according to the distribution
  over $k$-tuples induced by the tester, and set $y_i :=
  \widetilde{f_{u,v_i}}(x_i)$.
\item Return the constraint $\phi(y_1,\ldots,y_k) = 1$.  \een
\end{quote}
\vspace{.1pt}
\end{minipage}
\end{boxedminipage}
\end{center}

In step 3, $\pi_{u,v_i}$ is the permutation associated with the edge
$(u,v_i)\in E$, and $x\circ \pi_{u,v_i}$ is the string $x$ with its
coordinates permuted according to $\pi_{u,v_i}$.  For each $u\in V$,
it will be convenient for us to think of an assignment to the
corresponding $2^L$ variables of the CSP as a boolean function
$f_u:\bits^{L}\to\bits$, where $Z_{u,x} \leftarrow f_u(x)$; an
assignment to all $|V|\cdot 2^L$ variables can then be defined as a
set of $|V|$ boolean functions $\{ f_u:\bits^L\to\bits\}_{u\in V}$.  We
will assume that $G$ is regular; this is without loss of generality by
a result of Khot and Regev \cite{KR03}.

\subsection*{Completeness} 

Suppose $\opt(\Psi) \geq 1-\delta$, and let $\ell : V \to [L]$ be a
labelling achieving this.  Our goal is to exhibit an assignment to the
variables of the CSP $R(\Psi)$ that satisfies at least a $c-\eps$
fraction of constraints.  We consider the fraction satisfied by the
assignment $f_u(x) := \mathsf{DICT}_{\ell(u)}(x) = x_{\ell(u)}$ for
all $u\in V$.\medskip

Consider an edge $(u,v_i)\in E$ satisfied by $\ell$ ({\it i.e.}
$\ell(u) = \pi_{u,v_i}(\ell(v_i))$, and note that 
\[ \widetilde{f_{u,v_i}}(x) = f_{v_i}(x\circ \pi_{u,v_i}) = (x\circ
\pi_{u,v_i})_{\ell(v_i)} = x_{\pi_{u,v_i}(\ell(v_i))} = x_{\ell(u)} =
\mathsf{DICT}_{\ell(u)}.\] Therefore if the $k$ edges incident to $u$
(chosen in step 2 above) are all satisfied by $\ell$ then
$\E[\phi(\boldsymbol{y_1},\ldots,\boldsymbol{y_k})]$ is the probability
that the test accepts $\mathsf{DICT}_{\ell(u)}$, at least $c$ by our
assumption. Since $G$ is regular and $\ell$ satisfies a $1-\delta$
fraction of all edges, the probability that $k$ uniformly random edges
incident to a random $u\in V$ is satisfied by $\ell$ is at least
$1-k\delta$, and so we have $\opt(R(\Psi)) \geq c \cdot (1-k\delta)
\geq c -\eps$ (for sufficiently small $\delta$).

\subsection*{Soundness} 

We will assume that $\opt(R(\Psi)) \geq s + \eps$ and prove
$\opt(\Psi) = \Omega_\eps(1)$.  We first express the fraction of
satisfied constraints as
\[
s+ \eps \leq
\mathop{\Ex_{u,v_1,\ldots,v_k}}_{x_1,\ldots,x_k}\big[\phi\big(\widetilde{f_{u,v_1}}(x_1),\ldots,
\widetilde{f_{u,v_k}}(x_j)\big)\big] =
\Ex_{u,x_1,\ldots,x_k}\big[\phi\big(\Ex_{v_1}[\widetilde{f_{u,v_1}}(x_1)],\ldots,
\Ex_{v_k}[\widetilde{f_{u,v_k}}(x_k)]\big)\big] \] For each $u\in V$,
let $g_u\bits^L\to [-1,1]$ be the function defined by $g_u(x) :=
\E_{v\in N(u)}[\widetilde{f_{u,v}}(x)]$, and so the above can be
rewritten as $\E_{u,x_1,\ldots,x_n}[\phi(g_u(x_1),\ldots,g_u(x_k))]
\geq s + \eps$. By an averaging argument, at least an $\frac{\eps}{2}$
fraction of all $u\in V$ satisfy
$\E_{x_1,\ldots,x_k}[\phi(g_u(x_1),\ldots,g_u(x_k))]\geq s +
\frac{\eps}{2}$; call these $u\in V$ ``good''. By the soundness
condition of a dictator versus quasirandom test, it follows that if
$u$ is good $g_u$ cannot be $(\gamma,\gamma)$-quasirandom for some
$\gamma = \Omega_\eps(1)$; in particular, there must exist at least
one $i\in [L]$ such that $\Inf_i^{(1-\gamma)}(g_u)\geq \gamma$.
\medskip

Let $J_u = \{i\in [L] \stc \Inf_i^{(1-\gamma)}(g_u) \geq \gamma\}$,
and note that $|J_u| \leq \inv{\gamma^2}$ by Proposition
\ref{prop:bquasijuntas}.  We claim that every $i\in J_u$ satisfies
$\Pr_{v\in N(u)}\big[\Inf_{\pi^{-1}_{u,v}(i)}^{(1-\gamma)}(f_v) >
\frac{\gamma}{2}\big] > \frac{\gamma}{2}$ (once again, at least one
such $i$ exists if $u$ is good). To see this, it suffices to check
that
\[ \Ex_{v\in N(u)}\big[\Inf_{\pi_{u,v}(i)}^{(1-\gamma)}(f_v)\big] =
\Ex_{v}\big[\Inf_i^{(1-\gamma)}\big(\widetilde{f_{u,v}}\big)\big] \geq
\Inf_i^{(1-\gamma)}\big(\Ex_v\big[\widetilde{f_{u,v}}\big]\big) =
\Inf_i^{(1-\gamma)}(g_u) \geq \gamma; \] the claim then follows by
Markov's inequality.  For every $u\in V$ we also define $J_u' = \{j\in
[L] \stc \Inf_j^{(1-\gamma)}(f_u)\geq \frac{\gamma}{2}\}$, noting that
$|J_u'| \leq \frac{2}{\gamma^2}$. \medskip

Consider the following labelling $\ell: V \to [L]$: for each $u\in V$,
if $J_u\cup J_u'$ is non-empty we assign $u$ a uniformly random label
in $J_u\cup J_u'$, otherwise we assign $u$ an arbitrary label.  It
remains to prove that $\ell$ satisfies an $\Omega_\eps(1)$
fraction of edges.  We have shown that for at least an
$\frac{\eps}{2}\cdot\frac{\gamma}{2} = \frac{\eps\gamma}{4}$ fraction
of edges $(u,v)$ there exists an $i\in [L]$ such that $i\in J_u\cup
J_u'$ and $\pi_{u,v}^{-1}(i) \in J_v\cup J_v'$.  Conditioned on the
existence of such an $i$, the edge is satisfied if $\ell(u) = i$ and
$\ell(v) = \pi_{u,v}^{-1}(i)$ (recall that $(u,v)$ is satisfied if
$\ell(u) = \pi_{u,v}(\ell(v))$), and this happens with probability at
least $(|J_u\cup J_u'||J_v\cup J_v'|)^{-1} \geq
\big(\inv{\gamma^2}+\frac{2}{\gamma^2}\big)^2 = \frac{\gamma^4}{9}$. We
conclude that $\opt(\Psi) = \Omega(\eps\gamma\cdot \gamma^4) =
\Omega(\eps\gamma^5) = \Omega_\eps(1)$, and the proof is complete.

\bigskip

 For further details see \cite{Aus08}. 

\pagebreak

\end{document}